\documentclass[11pt,a4paper]{amsart}
\usepackage{fullpage}
\usepackage{booktabs} 

\usepackage{amsthm}
\usepackage{amsmath}
\usepackage{amsfonts}
\usepackage{amssymb}
\usepackage[bookmarks=true,hypertexnames=false]{hyperref}
\hypersetup{colorlinks=true, citecolor=blue, linkcolor=red, urlcolor=blue}

\newtheorem{theorem}{Theorem}[section]
\newtheorem{lemma}[theorem]{Lemma}

\newtheorem{definition}[theorem]{Definition}
\newtheorem{proposition}[theorem]{Proposition}
\newtheorem{claim}[theorem]{Claim}

\newcommand{\NP}{\textnormal{NP}}
\newcommand{\RP}{\textnormal{RP}}
\newcommand{\numP}{\#\textnormal{P}}
\newcommand{\trans}[4]{\ensuremath{\left[\begin{smallmatrix} #1 & #2 \\ #3 & #4 \end{smallmatrix}\right]}}

\def\critD{{\Delta_{c}}}
\def\critL{{\lambda_{c}}}
\def\critLint{{\lambda_{c}^{int}}}
\newcommand{\Spin}[3]{\#\textnormal{2}\textsc{Spin}$(#1,#2,#3)$}
\newcommand{\dSpin}[4]{\#$#4$-\textnormal{2}\textsc{Spin}$(#1,#2,#3)$}
\newcommand{\BIS}{\#\textsc{BIS}}
\newcommand{\SAT}{\#\textsc{Sat}}

\newcommand{\ceil}[1]{\ensuremath{\left\lceil #1 \right\rceil}}
\newcommand{\floor}[1]{\ensuremath{\left\lfloor #1 \right\rfloor}}
\newcommand{\abs}[1]{\ensuremath{\left| #1 \right|}}
\newcommand{\tree}[1]{\ensuremath{\mathbb{T}_{#1}}}
\let\epsilon\varepsilon
\let\phi\varphi

\newcommand{\deriv}[2]{\frac{\operatorname{d}{#1}}{\operatorname{d}{#2}}}
\newcommand{\pderiv}[2]{\frac{\partial{#1}}{\partial{#2}}}

\usepackage{tikz}
\usetikzlibrary{arrows,backgrounds,calc,fit,decorations.pathreplacing,decorations.markings,shapes.geometric}
\tikzset{every fit/.append style=text badly centered}
\usepackage{tikz-cd}

\newcommand{\ctn}[2]{\ensuremath{C_{#1,#2}}}
\newcommand{\symctn}[2]{\ensuremath{c_{#1,#2}}}

\newcommand{\bgl}{\ensuremath{(\beta,\gamma,\lambda)}~}
\def\RR{\mathbb{R}}

\makeatletter
\def\prob#1#2#3{\goodbreak\begin{list}{}{\labelwidth\z@ \itemindent-\leftmargin
                        \itemsep\z@  \topsep6\p@\@plus6\p@
                        \let\makelabel\descriptionlabel}
                      \item[\textbf{Name}]#1
                      \item[\textbf{Instance}]#2
                      \item[\textbf{Output}]#3
                \end{list}}
\makeatother

\title{Uniqueness, Spatial Mixing, and Approximation for Ferromagnetic 2-Spin Systems}

\author{Heng Guo}
\address[Heng Guo]{School of Informatics, University of Edinburgh, Informatics Forum, Edinburgh, EH8 9AB, United Kingdom.}
\email{hguo@inf.ed.ac.uk}

\author{Pinyan Lu}
\address[Pinyan Lu]{ITCS, Shanghai University of Finance and Economics, No.100
Wudong Road, Yangpu District, Shanghai, China.}
\email{lu.pinyan@mail.shufe.edu.cn}
\date{\today}

\begin{document}

\begin{abstract}
  We give fully polynomial-time approximation schemes (FPTAS) for the partition function of ferromagnetic 2-spin systems in certain parameter regimes. The threshold we obtain is almost tight up to an integrality gap. Our technique is based on the correlation decay framework. The main technical contribution is a new potential function, with which we establish a new kind of spatial mixing.
\end{abstract}

\maketitle

\section{Introduction}

Spin systems model nearest neighbor interactions.
In this paper we study $2$-state spin systems.
An instance is a graph $G=(V,E)$, and a configuration $\sigma: V \rightarrow \{0,1\}$ assigns one of the two spins ``0'' and ``1'' to each vertex.
The local interaction along an edge is specified by a matrix $\mathbf{A}=\trans{A_{0,0}}{A_{0,1}}{A_{1,0}}{A_{1,1}}$,
where $A_{i,j}$ is the (non-negative) local weight when the two endpoints are assigned $i$ and $j$ respectively.
We study symmetric edge interactions, that is, $A_{0,1}=A_{1,0}$.
Normalize $\mathbf{A}$ so that $\mathbf{A}=\trans{\beta}{1}{1}{\gamma}$.
Moreover, we also consider the external field, specified by a mapping $\pi:V\to\mathbb{R}^+$.
When a vertex is assigned ``0'', we give it a weight $\pi(v)$.
For a particular configuration $\sigma$,
its weight $w(\sigma)$ is a product over all edge interactions and vertex weights, that is
\begin{align*}
  w(\sigma) = \beta^{m_0(\sigma)} \gamma ^{m_1(\sigma)} \prod_{v\mid \sigma(v)=0}\pi(v),
\end{align*}
where $m_0(\sigma)$ is the number of $(0,0)$ edges under the configuration $\sigma$ and $m_1(\sigma)$ is the number of $(1,1)$ edges.
An important special case is the Ising model, where $\beta=\gamma$.
The Gibbs measure is a natural distribution in which each configuration $\sigma$ is drawn with probability proportional to its weight,
that is, $\Pr_{G;\beta,\gamma,\pi}(\sigma)\sim w(\sigma)$.
The normalizing factor of the Gibbs measure is called the partition function,
defined by $Z_{\beta,\gamma,\pi}(G) = \sum_{\sigma:V\rightarrow\{0,1\}}w(\sigma)$.
The partition function encodes rich information regarding the macroscopic behavior of the spin system.
We will be interested in the computational complexity of approximating $Z_{\beta,\gamma,\pi}(G)$.
We also simply write $Z_{\beta,\gamma,\lambda}(G)$ when the field is uniform, that is, $\pi(v)=\lambda$ for all $v\in V$.
A system with uniform fields is specified by the three parameters $(\beta,\gamma,\lambda)$.

Spin systems not only are interesting in statistical physics,
but also find applications in computer science, under the name of \emph{Markov random fields}.
Indeed, a $2$-state spin system is equivalent to a binary Markov random field.
For example, Boltzmann Machines \cite{AHS85} can be viewed as a special case where $\gamma=1$.
Computing the partition function is a central task in statistical inference.
According to their physical and computational properties,
spin systems can be classified into two families:
\emph{ferromagnetic} systems where the edge interaction is attractive ($\beta \gamma >1$),
and \emph{anti-ferromagnetic} systems where it is repulsive ($\beta \gamma < 1$).

Recently, beautiful connections have been established regarding three different properties of anti-ferromagnetic $2$-spin systems:
tree uniqueness, spatial mixing, and computational complexity transitions.
The uniqueness of Gibbs measures in infinite regular trees\footnote{This property is called ``tree uniqueness'' or ``uniqueness'' for short. See Sec \ref{sec:uniqueness} and \ref{sec:missing-proofs:uniqueness} for details.} of degrees up to $\Delta$
implies correlation decay\footnote{That is, the correlation of any two vertices decay exponentially in distance. This is also known as ``spatial mixing''.} 
in all graphs of maximum degree $\Delta$,
and therefore the existence of fully polynomial-time approximation scheme (FPTAS) for the partition function \cite{Wei06,LLY12,SST14,LLY13}.
On the other hand, if the tree uniqueness fails, then long range correlation appears and the partition function has no
fully polynomial-time randomized approximation scheme (FPRAS) unless $\NP=\RP$ \cite{Sly10, SS14, GSV16}.
It suggests that tree uniqueness, spatial mixing,
and the computational complexity of approximating the partition function, line up perfectly in the anti-ferromagnetic regime.

For ferromagnetic systems, the picture is much less clear.
In a seminal paper~\cite{JS93}, Jerrum and Sinclair gave an FPRAS
for the ferromagnetic Ising model $\beta=\gamma >1$ with any consistent external field for general graphs without degree bounds.
Thus, there is no computational complexity transition of approximating these models,
whereas uniqueness and spatial mixing do exhibit phase transition.
This is in sharp contrast to anti-ferromagnetic Ising models $\beta=\gamma<1$,
where computational and phase transitions align perfectly.
It is not clear at all whether spatial mixing or correlation decay plays any role in the computational complexity.

For more general ferromagnetic $2$-spin systems with external fields, the threshold of efficiently approximating the partition function is still open.
On the complexity side, Goldberg and Jerrum showed that any ferromagnetic $2$-spin system is no harder than 
counting independent sets in bipartite graphs (\BIS) \cite{GJ07}, which is conjectured to have no FPRAS \cite{DGGJ03}
(the approximation complexity of \BIS\ is still open).
Based on an earlier result \cite{CGGGJSV16},
Liu, Lu and Zhang showed that approximating the partition function is \BIS-hard
if we allow external fields beyond $\left({\gamma}/{\beta}\right)^{\frac{\floor{\Delta_c}+2}{2}}$
where $\Delta_c = \frac{\sqrt{\beta\gamma}+1}{\sqrt{\beta\gamma}-1}$~\cite{LLZ14a}.\footnote{Here and below we assume $\beta\le\gamma$ due to symmetry.}

On the algorithmic side, by reducing to the Ising model, 
an MCMC based FPRAS is known for the range of $\lambda\le{\gamma}/{\beta}$~\cite{LLZ14a} (improving upon \cite{GJP03}).
On the other hand, if we apply the correlation decay algorithmic framework to various pairs of parameters $(\beta,\gamma)$,
it is not hard to get bounds better than $\gamma/\beta$.
However, such success for individual problems does not seem to share meaningful inner connections.
In particular, it is not clear how far one can push this method,
and to the best of our knowledge, no threshold has even been conjectured.

\subsection{Our Contribution}

In this paper, we identify a new threshold that almost tightly maps out the boundary of the correlation decay regime,
that is, $\critL:=\left({\gamma}/{\beta}\right)^{\frac{\critD+1}{2}}=\left({\gamma}/{\beta}\right)^{\frac{\sqrt{\beta\gamma}}{\sqrt{\beta\gamma}-1}}$.
We show that for any $\lambda<\critL$ a variant of spatial mixing holds (Theorem \ref{thm:mixing}) for arbitrary trees.
An interesting feature of our work is that we do not restrict the degree or the shape of the tree.
This is almost tight since it does not hold if $\lambda>\left({\gamma}/{\beta}\right)^{\frac{\ceil{\Delta_c}+1}{2}}$.
This spatial mixing is weaker than what an algorithm usually requires, 
but in the regime of $\beta\le 1$ it implies (and therefore is equivalent to) strong spatial mixing.
As an algorithmic consequence, we have FPTAS for all $\beta\le 1<\gamma$, $\beta\gamma>1$, and $\lambda<\critL$ (Theorem \ref{thm:algorithm}).
Recall that if we allow $\lambda$ beyond $\left({\gamma}/{\beta}\right)^{\frac{\floor{\Delta_c}+2}{2}}$, then the problem is \BIS-hard \cite{LLZ14a}.
Hence only an integral gap remains for the $\beta\le 1<\gamma$ case.

Formally, let $p_v$ be the marginal probability of $v$ being assigned ``0''.
\begin{theorem} \label{thm:mixing}
  Let $(\beta,\gamma,\lambda)$ be a set of parameters of the system such that $\beta\gamma>1$, $\beta\le\gamma$, and $\lambda<\critL$.
  Let $T_v$ and $T'_{v'}$ be two trees with roots $v$ and $v'$ respectively.
  If the two trees have the same structure in the first $\ell$ levels, then $|p_v-p_{v'}|\leq O(\exp(-\ell))$.
\end{theorem}

In other words, if we simply truncate a tree at depth $\ell$, the marginal probability of its root will change by only at most $O(\exp(-\ell))$.
Surprisingly, if we replace $\critL$ by its integral counterpart,
then this implication no longer holds and there is a counterexample (see Section \ref{sec:hardness}).
More precisely, it is no longer true that the uniqueness in infinite regular trees implies correlation decay in graphs or even trees,
since our counterexample is an irregular tree.
We note that this is in sharp contrast to anti-ferromagnetic systems, where (integral) uniqueness implies correlation decay.

From the computational complexity point of view,
we would like to get an FPTAS for the partition function, which requires a condition called strong spatial mixing (SSM).
It is stronger than the spatial mixing established in Theorem \ref{thm:mixing} by imposing arbitrary partial configurations.
We are able to prove SSM with $\lambda<\critL$ in the range of $\beta\leq 1$.
Indeed, if $\beta\leq 1$, then the two versions of spatial mixing are equivalent. 
Let $I$ be an interval of the form $[\lambda_1,\lambda_2]$ or $(\lambda_1,\lambda_2]$.
We consider the following problem.

\prob{\Spin{\beta}{\gamma}{I}}{A graph $G=(V,E)$ and a mapping $\pi:V\to\mathbb{R}^+$, such that $\pi(v)\in I$ for any $v\in V$.}{$Z_{\beta,\gamma,\pi}(G)$.}

Then we have the following theorem.

\begin{theorem} \label{thm:algorithm}
  Let $(\beta,\gamma,\lambda)$ be a set of parameters of the system such that $\beta\gamma>1$, $\beta\leq 1$ and $\lambda<\critL$.
  There is an FPTAS for \Spin{\beta}{\gamma}{(0,\lambda]}.
\end{theorem}
Therefore, we get an almost tight dichotomy for ferromagnetic $2$-spin systems when $\beta\le 1$,
since \Spin{\beta}{\gamma}{(0,\lambda]} is \BIS-hard, if $\lambda$ is larger than the integral counterpart of $\critL$
\cite{LLZ14a} (see also Proposition \ref{prop:hardness}).

The reason behind $\critL$ is a nice interplay among uniqueness, spatial mixing, and approximability.
We start with some purely mathematical observations on the symmetric tree recursion $f_d(x)=\lambda\left( \frac{\beta x+1}{x+\gamma} \right)^d$,
an increasing function in $x$.
Relax the range of $d$ in $f_d(x)$ to be real numbers.
Then $\critD$ is the critical (possibly fractional) degree
and $\critL$ is the corresponding critical external field for the recursion to have a unique fixed point.
This set of critical parameters enjoys some very nice mathematical properties.
For $d=\critD$ and $\lambda =\critL$,
the function $f_d(x)$ has a unique fixed point $\widehat{x}=\sqrt{\gamma/\beta}$ and $f'_d(\widehat{x})=1$.
Moreover, it also satisfies that $f''_d(\widehat{x})=0$,
which is a necessary condition for the contraction of the tree recursion.
(This is easy to derive using the heuristic of finding potential functions described in \cite{LLY13}.)
All these nice mathematical properties prove to be useful in our later analysis.
For degrees other than $\critD$, their critical external fields are much less convenient --- 
the function $f_d(x)$ has two fixed points: one is crossing and the other is tangent.
Moreover, $f''_d(\hat{x})=0$ does not necessarily hold.

The proof of Theorem \ref{thm:mixing} uses the potential method to analyze decay of correlation,
which is now streamlined (see e.g.\ \cite{LLY13, SSSY17}).
The main difficulty is to find a good potential function.
In other words, we want to solve a variational problem minimizing the maximum of the decay rate function.
The main novelty in our solution is that we restrict variables to the range of $(0,\frac{\lambda}{1+\lambda}]$ 
and our potential function is well-defined \emph{only} in this range.
This is in fact necessary, as otherwise the statement does not hold, and is valid for the setting of Theorem \ref{thm:mixing}.
Our choice leads to a relatively clean and significantly simpler proof, comparing to similar proofs in other settings.
In particular, we do not need the ``symmetrization'' argument (see e.g.\ \cite{LLY13,SSSY17}).
We also use a trick of truncating the potential to deal with unbounded degrees (see Eq.\ \eqref{eqn:phi2:definition2}).

For the range of $\beta>1$, SSM does not hold even if $\lambda<\critL$.
However, we conjecture that Theorem \ref{thm:algorithm} can be extended to the $\beta>1$ range as well,
mainly due to Theorem \ref{thm:mixing}, which does not require $\beta\le1$.
Moreover, we show that even if $\beta>1$, the marginal probability in any instance is within the range of $(0,\frac{\lambda}{1+\lambda}]$ given $\lambda<\critL$
(see Proposition \ref{prop:bounded}).
This seems to imply that the main reason why our algorithm fails is due to pinnings (forcing a vertex to be ``0'' or ``1'') in the self-avoiding walk tree construction,
whereas in a real instance these pinnings cannot aggregate enough ``bad'' influence.
However, to turn such intuition into an algorithm requires a careful treatment of these pinnings to achieve an FPTAS without SSM.
We leave this as an important open problem.

At last, we note that neither $\critL$ nor its integral counterpart is the exact threshold in each own respect, even if $\beta\le 1$.
Strong spatial mixing continues to hold even if $\lambda>\critL$ in a small interval.
We give a concrete example to illustrate this fact in Section \ref{sec:beyond-critL}, Proposition \ref{prop:beyond-critL}.
Moreover, as mentioned earlier, an irregular tree exists where the correlation decay threshold is lower than the threshold for all infinite regular trees.
This is discussed in Section \ref{sec:hardness}.
It is another important open question to figure out the exact threshold between $\critL$ and its integral counterpart(s).

\section{Preliminaries}

An instance of a $2$-spin system is a graph $G=(V,E)$.
A configuration $\sigma:V \rightarrow \{0,1\}$ assigns one of the two spins ``0'' and ``1'' to each vertex.
We normalize the edge interaction to be $\trans{\beta}{1}{1}{\gamma}$,
and also consider the external field, specified by a mapping $\pi:V\to\mathbb{R}^+$.
When a vertex is assigned ``0'', we give it a weight $\pi(v)$.
All parameters are non-negative.
For a particular configuration $\sigma$,
its weight $w(\sigma)$ is a product over all edge interactions and vertex weights, that is
\begin{align} \label{eqn:weight}
  w(\sigma) = \beta^{m_0(\sigma)} \gamma ^{m_1(\sigma)} \prod_{v\mid \sigma(v)=1}\pi(v),
\end{align}
where $m_0(\sigma)$ is the number of $(0,0)$ edges given by the configuration $\sigma$ and $m_1(\sigma)$ is the number of $(1,1)$ edges.
An important special case is the Ising model, where $\beta=\gamma$.
Notice that in the statistic physics literature,
parameters are usually chosen to be the logarithms of our parameters above.
Different parameterizations do not affect the complexity of the same system.

We also write $\lambda_v:=\pi(v)$.
If $\pi$ is a constant function such that $\lambda_v=\lambda>0$ for all $v\in V$, we also denote it by $\lambda$.
We say $\pi$ has a lower bound (or an upper bound) $\lambda>0$,
if $\pi$ satisfies the guarantee that $\lambda_v\ge \lambda$ (or $\lambda_v\le\lambda$).

The Gibbs measure is a natural distribution in which each configuration $\sigma$ is drawn with probability proportional to its weight,
that is, 
\begin{align}  \label{eqn:Gibbs}
  \Pr_{G;\beta,\gamma,\pi}(\sigma)\propto w(\sigma).
\end{align}
The normalizing factor of the Gibbs measure is called the partition function, defined by $Z_{\beta,\gamma,\pi}(G) = \sum_{\sigma:V\rightarrow\{0,1\}}w(\sigma)$.
Recall that we are interested in the computational problem \Spin{\beta}{\gamma}{I},
where $I$ is an interval of the form $[\lambda_1,\lambda_2]$ or $(\lambda_1,\lambda_2]$, for which $Z_{\beta,\gamma,\pi}(G)$ is the output.
When input graphs are restricted to have a degree bound $\Delta$,
we write \dSpin{\beta}{\gamma}{I}{\Delta} to denote the problem.
When the field is uniform, that is, $\lambda$ is the only element in $I$, we simply write \Spin{\beta}{\gamma}{\lambda}.
Due to \cite{CK12} and a standard diagonal transformation,
for any constant $\lambda>0$, \Spin{\beta}{\gamma}{\lambda} is \numP-hard unless $\beta=\gamma=0$ or $\beta\gamma=1$.

\subsection{The Self-Avoiding Walk Tree}

We briefly describe Weitz's algorithm \cite{Wei06}.
Our algorithms presented later will follow roughly the same paradigm.

The Gibbs measure defines a marginal distribution of spins for each vertex.
Let $p_v$ denote the probability of a vertex $v$ being assigned ``0''.
Since the system is self-reducible, \Spin{\beta}{\gamma}{\lambda} is equivalent to computing $p_v$ for any vertex $v$ \cite{JVV86}
(for details, see for example Lemma \ref{lem:potential:algorithm}).

Let $\sigma_{\Lambda}\in\{0,1\}^\Lambda$ be a configuration of $\Lambda\subset V$.
We call vertices in $\Lambda$ \emph{fixed} and other vertices \emph{free}.
We use $p_v^{\sigma_{\Lambda}}$ to denote the marginal probability of $v$ being assigned ``0''
conditional on the configuration $\sigma_{\Lambda}$ of $\Lambda$.

Suppose the instance is a tree $T$ with root $v$.
Let $R_T^{\sigma_\Lambda}:=p_v^{\sigma_\Lambda}/(1-p_v^{\sigma_\Lambda})$ be the ratio between the two probabilities
that the root $v$ is $0$ and $1$, while imposing some condition $\sigma_\Lambda$
(with the convention that $R_T^{\sigma_\Lambda}=\infty$ when $p_v^{\sigma_\Lambda}=1$).
Suppose that $v$ has $d$ children $v_i,\dots v_d$.
Let $T_i$ be the subtree with root $v_i$.
Due to the independence of subtrees,
it is straightforward to get the following recursion for calculating $R_T^{\sigma_\Lambda}$:
\begin{align}\label{eqn:tree-recursion}
  R^{\sigma_\Lambda}_T &=F_d\left(R_{T_1}^{\sigma_\Lambda},\ldots,R_{T_d}^{\sigma_\Lambda}\right),
\end{align}
where the function $F_d(x_1,\dots,x_d)$ is defined as
\begin{align*}
  F_d(x_1,\dots,x_d)&:=\lambda_v\prod_{i=1}^d\frac{\beta x_i+1}{x_i+\gamma}.
\end{align*}
We allow $x_i$'s to take the value $\infty$
as in that case the function $F_d$ is clearly well defined.
In general we use capital letters like $F,G,C,\dots$ to denote multivariate functions,
and small letters $f,g,c,\dots$ to denote their symmetric versions, where all variables take the same value.
Here we define $f_d(x):=\lambda\left( \frac{\beta x+1}{x+\gamma} \right)^d$
to be the symmetric version of $F_d(\mathbf{x})$ obtained by setting $x_1=\dots=x_d=x$.

Let $G(V,E)$ be a graph.
Similarly define $R_{G,v}^{\sigma_\Lambda}:=p_v^{\sigma_\Lambda}/(1-p_v^{\sigma_\Lambda})$.
In contrast to the case of trees,
there is no easy recursion to calculate $R_{G,v}^{\sigma_\Lambda}$ for a general graph $G$.
This is because of dependencies introduced by cycles.
Weitz \cite{Wei06} reduced computing the marginal distribution of $v$ in a general graph $G$ to
that in a tree, called the self-avoiding walk (SAW) tree, denoted by $T_{\text{SAW}}(G,v)$.
To be specific, given a graph $G=(V,E)$ and a vertex $v\in V$,
$T_{\text{SAW}}(G,v)$ is a tree with root $v$ that
enumerates all self-avoiding walks originating from $v$ in $G$,
with additional vertices closing cycles as leaves of the tree.
Each vertex in the new vertex set $V_{\text{SAW}}$ of $T_{\text{SAW}}(G,v)$ corresponds to a vertex in $G$,
but a vertex in $G$ may be mapped to more than one vertices in $V_{\text{SAW}}$.
A boundary condition is imposed on leaves in $V_{\text{SAW}}$ that close cycles.
The imposed colors of such leaves depend on whether the cycle is formed from a small vertex to a large vertex or conversely,
where the ordering is arbitrarily chosen in $G$.
Vertex sets $S\subset \Lambda\subset V$ are mapped to respectively $S_{\text{SAW}}\subset\Lambda_{\text{SAW}}\subset V_{\text{SAW}}$,
and any configuration $\sigma_\Lambda\in\{0,1\}^\Lambda$ is mapped to $\sigma_{\Lambda_{\text{SAW}}}\in\{0,1\}^{\Lambda_{\text{SAW}}}$.
With slight abuse of notations we may write $S=S_{\text{SAW}}$ and $\sigma_\Lambda=\sigma_{\Lambda_{\text{SAW}}}$
when no ambiguity is caused.

\begin{proposition}[Theorem 3.1 of Weitz~\cite{Wei06}]\label{prop:SAW}
  Let $G=(V,E)$ be a graph, $v\in V$,
  $\sigma_\Lambda\in\{0,1\}^\Lambda$ be a configuration on $\Lambda\subset V$,
  and $S\subset V$.
  Let $T=T_{\mathrm{SAW}}(G,v)$ be constructed as above.
  It holds that
  \begin{align*}
    R_{G,v}^{\sigma_\Lambda}=R_T^{\sigma_\Lambda}.
  \end{align*}
  Moreover, the maximum degree of $T$ is at most the maximum degree of $G$,
  $\mathrm{dist}_G(v,S)=\mathrm{dist}_T(v,S_{\text{SAW}})$,
  and any neighborhood of $v$ in $T$ can be constructed in time proportional to the size of the neighborhood.
\end{proposition}

The SAW tree construction does not solve a \numP-hard problem,
since $T_{\mathrm{SAW}}(G,v)$ is potentially exponentially large in size of $G$.
For a polynomial time approximation algorithm, we may run the tree recursion within some polynomial size, or equivalently a logarithmic depth.
At the boundary where we stop, we plug in some arbitrary values.
The question is then how large is the error due to our random guess.
To guarantee the performance of the algorithm, we need the following notion of \emph{strong spatial mixing}.

\begin{definition}\label{def:correlation-decay}
  A spin system on a family $\mathcal G$ of graphs is said to exhibit \emph{strong spatial mixing} $($SSM$)$ if for any graph $G=(V,E)\in\mathcal{G}$,
  any $v\in V,\Lambda\subset V$ and any $\sigma_{\Lambda},\tau_{\Lambda}\in\{0,1\}^{\Lambda}$,
  \begin{align}\label{eqn:ssm}
    \abs{p_v^{\sigma_{\Lambda}}-p_v^{\tau_{\Lambda}}}\leq \exp(-\Omega(\mathrm{dist}(v,S))),
  \end{align}
  where $S\subset\Lambda$ is the subset on which $\sigma_{\Lambda}$ and $\tau_{\Lambda}$ differ,
  and $\mathrm{dist}(v,S)$ is the shortest distance from $v$ to any vertex in $S$.
\end{definition}

\emph{Weak spatial mixing} is defined similarly by replacing $\mathrm{dist}(v,S)$ with $\mathrm{dist}(v,\Lambda)$ in \eqref{eqn:ssm},
and it corresponds to the uniqueness condition introduced below in Section \ref{sec:uniqueness}.
Spatial mixing properties are also called correlation decay in statistical physics.

If SSM holds, then the error caused by early termination in $T_{\mathrm{SAW}}(G,v)$ and arbitrary boundary values is only exponentially small in the depth.
Hence the algorithm is an FPTAS.
In a lot of cases, the existence of an FPTAS boils down to establish SSM.

\subsection{The Uniqueness Condition in Regular Trees}\label{sec:uniqueness}

Let \tree{d}\ denote the infinite $d$-regular tree, also known as the \emph{Bethe lattice} or the \emph{Cayley tree}.
If we pick an arbitrary vertex as the root of \tree{d},
then the root has $d$ children and every other vertex has $d-1$ children.
Notice that the difference between \tree{d}\ and an infinite $(d-1)$-ary tree is only the degree of the root.
We assume in this subsection that the field is uniform $\lambda>0$.

The Gibbs measure in $\tree{d}$ is a probability measure where conditioned on any arbitrary configuration on the boundary of a finite set $S$,
the resulting distribution on $S$ is the same as the one given by \eqref{eqn:Gibbs} on $S$ with the same boundary condition.
When the Gibbs measure is unique, we say that \emph{the uniqueness condition} holds in $\tree{d}$.
The tree recursion \eqref{eqn:tree-recursion} turns out to be important in analyzing the uniqueness condition.
Due to the symmetric structure of \tree{d},
it becomes $R_v=f_{d-1}(R_{v_i})$ (for any vertex $v$ other than the root),
where $f_d(x)=\lambda\left(\frac{\beta x+1}{x+\gamma}\right)^d$ is the symmetrized version of $F_d(\mathbf{x})$.

For anti-ferromagnetic systems, that is, $\beta\gamma<1$, there is a unique fixed point to $f_d(x)=x$,
denoted by $\widehat{x}$.
It has been shown that the Gibbs measure in \tree{d}\ is unique if and only if
$\abs{f_{d-1}'(\widehat{x})}\le1$ \cite{Kel85,Geo11}.

In contrast, if $\beta\gamma>1$, then there may be $1$ or $3$ positive fixed points such that $x=f_d(x)$. 
It is known \cite{Kel85,Geo11} that the Gibbs measure of two-state spin systems in \tree{d}\ is unique if and only if
there is only one fixed point for $x=f_{d-1}(x)$,
or equivalently, for all fixed points $\widehat{x}_d$ of $f_d(x)$, $f_d'(\widehat{x}_d)<1$.

Let $\critD:=\frac{\sqrt{\beta\gamma}+1}{\sqrt{\beta\gamma}-1}$.
Then we have the following result.

\begin{proposition}
  If $\Delta-1<\critD$,
  then the uniqueness condition in $\tree{\Delta}$ holds regardless of the field.
  \label{prop:uniqueness:bounded}
\end{proposition}

Note that the condition $\Delta-1<\critD$ matches the exact threshold of fast mixing for Glauber dynamics in the Ising model \cite{MS13}.
In Section \ref{sec:algorithm:bounded}, we will show that, SSM holds and there exists an FPTAS for the partition function,
in graphs with degree bound $\Delta<\critD+1$ and any field $\lambda>0$.
This is Theorem \ref{thm:tractable:first-order}.

To study general graphs, one needs to consider infinite regular trees of all degrees.
If $\beta>1$ (still assuming $\beta\gamma>1$ and $\beta\le\gamma$),
then there is no $\lambda$ such that the uniqueness condition holds in \tree{d} for all degrees $d\ge 2$.
In contrast, let $\critLint:=(\gamma/\beta)^{\frac{\ceil{\critD}+1}{2}}$ and for $\beta\le 1$, we have the following.

\begin{proposition}\label{prop:uniqueness:general}
  Let $(\beta,\gamma)$ be two parameters such that $\beta\gamma>1$.
  \begin{itemize}
    \item If $\beta\le 1<\gamma$, then uniqueness holds for \tree{d} with all degrees $d\ge 2$ if and only if $\lambda<\critLint$.
    \item If $\beta,\gamma >1$, then there is no $\lambda>0$ such that uniqueness holds for all \tree{d} with $d\ge2$.
  \end{itemize}
\end{proposition}

However, there exist \bgl and an (irregular) tree $T$ such that $\beta\gamma>1$, $\beta\le 1<\gamma$, and $\lambda<\critLint$
and SSM does not hold in $T$.
This is discussed in Section \ref{sec:hardness}.
Recall that $\critL:=(\gamma/\beta)^{\frac{\critD+1}{2}}$.
If we replace $\critLint$ with $\critL\le \critLint$ in the condition of Proposition~\ref{prop:uniqueness:general},
that is, $\beta\gamma>1$, $\beta\le 1<\gamma$, and $\lambda<\critL$, then SSM holds in all graphs and an FPTAS exists.
This is shown in Section~\ref{sec:algorithm:general}, Theorem~\ref{thm:tractable:beta<1}.

Details and proofs about Propositions \ref{prop:uniqueness:bounded} and \ref{prop:uniqueness:general} are given in Section \ref{sec:missing-proofs:uniqueness}.

\subsection{The Potential Method} \label{sec:potential}

We would like to prove the strong spatial mixing in arbitrary trees, sometimes with bounded degree $\Delta$, under certain conditions.
This is sufficient for approximation algorithms due to the self-avoiding walk tree construction.
Our main technique in the analysis is the potential method.
The analysis in this section is a standard routine, with some specialization to ferromagnetic 2-spin models (cf.\ \cite{LLY13,SSSY17}).
To avoid interrupting the flow, we move all details and proofs to Section \ref{sec:missing-proofs}.

Roughly speaking, instead of studying \eqref{eqn:tree-recursion} directly,
we use a potential function $\Phi(x)$ to map the original recursion to a new domain
(see the commutative diagram Figure \ref{fig:commutative}).
Morally we can choose whatever function as the potential function.
However, we would like to pick ``good'' ones so as to help the analysis of the contraction.
Define $\varphi(x):=\Phi'(x)$ and
\begin{align*}
  \ctn{\varphi}{d}(\mathbf{x}) := \varphi(F_d(\mathbf{x}))\cdot \sum_{i=1}^d\abs{\frac{\partial F_d}{\partial x_i}}\frac{1}{\varphi(x_i)}.
\end{align*}

\begin{definition} \label{def:potential}
  Let $\Phi:\RR^+\rightarrow\RR^+$ be a differentiable and monotonically increasing function.
  Let $\varphi(x)$ and $\ctn{\varphi}{d}(\mathbf{x})$ be defined as above.
  Then $\Phi(x)$ is a \emph{good potential function for degree $d$ and field $\lambda$} if it satisfies the following conditions:
  \begin{enumerate}
    \item \label{cond:boundedness} there exists a constant $C_1,C_2>0$
      such that $C_1\le\varphi(x)\le C_2$ for all $x\in[\lambda\gamma^{-d},\lambda\beta^{d}]$;
    \item \label{cond:contraction} there exists a constant $\alpha<1$
      such that $\ctn{\varphi}{d}(\mathbf{x})\le\alpha$ for all $x_i\in[\lambda\gamma^{-d},\lambda\beta^d]$.
  \end{enumerate}
\end{definition}
We say $\Phi(x)$ is a good potential function for $d$ and field $\pi$,
if $\Phi(x)$ is a good potential function for $d$ and any $\lambda$ in the codomain of $\pi$,

In Definition \ref{def:potential}, Condition \ref{cond:boundedness} is rather easy to satisfy.
The crux is in fact Condition \ref{cond:contraction}.
We call $\alpha$ in Condition \ref{cond:contraction} the amortized contraction ratio of $\Phi(x)$.
It has the following algorithmic implication.
The proof is based on establishing strong spatial mixing.

\begin{lemma}\label{lem:potential:algorithm}
  Let $(\beta,\gamma)$ be two parameters such that $\beta\gamma>1$.
  Let $G=(V,E)$ be a graph with a maximum degree $\Delta$ and $n$ many vertices and $\pi$ be a field on $G$.
  Let $\lambda=\max_{v\in V}\{\pi(v)\}$.
  If there exists a good potential function for $\pi$ and all $d\in[1, \Delta-1]$ with contraction ratio $\alpha<1$,
  then $Z_{\beta,\gamma,\pi}(G)$ can be approximated deterministically within a relative error $\epsilon$
  in time $O\left(n\left( \frac{n\lambda}{\epsilon} \right)^{\frac{\log(\Delta-1)}{-\log\alpha}}\right)$.
\end{lemma}

When the degree is unbounded, the SAW tree may grow super polynomially even if the depth is of order $\log n$.
We use a refined metric replacing the naive graph distance used in Definition \ref{def:correlation-decay}.
Strong spatial mixing under this metric is also called \emph{computationally efficient correlation decay} \cite{LLY12,LLY13}.

\begin{definition} \label{def:Mbased-depth}
  Let $T$ be a rooted tree and $M>1$ be a constant.
  For any vertex $v$ in $T$, define the \emph{$M$-based depth} of $v$,
  denoted $\ell_M(v)$, such that $\ell_M(v)=0$ if $v$ is the root,
  and $\ell_M(v)=\ell_M(u)+\lceil\log_M(d+1)\rceil$
  if $v$ is a child of $u$ and $u$ has degree $d$.
\end{definition}

Let $B(\ell)$ be the set of all vertices whose $M$-based depths of $v$ is at most $\ell$.
It is easy to verify inductively such that $|B(\ell)|\le M^\ell$ in a tree.
We then define a slightly stronger notion of potential functions.

\begin{definition}\label{def:potential:universal}
  Let $\Phi:\RR^+\rightarrow\RR^+$ be a differentiable and monotonically increasing function.
  Let $\varphi(x)$ and $\ctn{\varphi}{d}(\mathbf{x})$ defined in the same way as in Definition \ref{def:potential}.
  Then $\Phi(x)$ is a \emph{universal potential function} for the field $\lambda$ if it satisfies the following conditions:
  \begin{enumerate}
    \item \label{cond:universal-boundedness} there are two constants $C_1,C_2>0$ such that $C_1 \le \varphi(x) \le C_2$ for any $x\in(0,\lambda]$;
    \item \label{cond:universal-contraction}
    there exists a constant $\alpha<1$ such that for all $d$, $\ctn{\varphi}{d}(\mathbf{x})\le\alpha^{\ceil{\log_M (d+1)}}$ for all $x_i\in(0,\lambda]$;
  \end{enumerate}
\end{definition}

We say $\Phi(x)$ is a universal potential function for a field $\pi$,
if $\Phi(x)$ is a universal potential function for any $\lambda$ in the codomain of $\pi$.
We also call $\alpha$ the contraction ratio and call $M$ the base.
The following two lemmas show that our main theorems follow from the existence of a universal potential function.

The way we define universal potential functions restricts them to only apply to the range of $(0,\lambda]$.
This will be true in our applications (see for example Claim \ref{claim:bound}).

\begin{lemma}
  Let $(\beta,\gamma,\lambda)$ be three parameters such that $\beta\gamma>1$, $\beta\le\gamma$, and $\lambda<\critL$.
  Let $T$ and $T'$ be two trees that agree on the first $\ell$ levels with root $v$ and $v'$ respectively.
  If there exists a universal potential function $\Phi(x)$,
  then $\abs{p_v-p_{v'}}\le O(\exp(-\ell))$.
  \label{lem:potential:spatial-mixing}
\end{lemma}

\begin{lemma}\label{lem:potential:algorithm:universal}
  Let $(\beta,\gamma)$ be two parameters such that $\beta\gamma>1$ and $\beta\le 1 <\gamma$.
  Let $G=(V,E)$ be a graph with $n$ many vertices and $\pi$ be a field on $G$.
  Let $\lambda=\max_{v\in V}\{\pi(v)\}$.
  If there exists a universal potential function $\Phi(x)$ for $\pi$ with contraction ratio $\alpha<1$ and base $M$,
  then $Z_{\beta,\gamma,\pi}(G)$ can be approximated deterministically within a relative error $\epsilon$
  in time $O\left(n^3\left( \frac{n\lambda}{\epsilon} \right)^{\frac{\log M}{-\log\alpha}}\right)$.
\end{lemma}

\section{Correlation Decay below \texorpdfstring{$\critD$}{the Critical Degree} or \texorpdfstring{$\critL$}{the Critical Field}} \label{sec:algorithm}

In this section, we show our main results.
We will first show a folklore result for bounded degree graphs with a very simple proof.
Then we continue to show the main theorem regarding general graphs.
We carefully choose two appropriate potential functions
and then apply Lemma \ref{lem:potential:algorithm} or Lemma \ref{lem:potential:algorithm:universal}.

\subsection{Bounded Degree Graphs} \label{sec:algorithm:bounded}

We first apply our framework to get FPTAS for graphs with degree bound $\Delta<\critD+1=\tfrac{2\sqrt{\beta\gamma}}{\sqrt{\beta\gamma}-1}$.
Correlation decay for graphs with such degree bounds is folklore and can be found in \cite{Lyo89} for the Ising model.
Algorithmic implications are also shown, e.g.\ in \cite{ZLB11}.
As we shall see, the proof is very simple in our framework.
Note that $\lambda$, $\Delta$, and $\alpha$ are considered constants for the FPTAS.

\begin{theorem}
  Let $(\beta,\gamma)$ be two parameters such that $\beta\gamma>1$.
  Let $G=(V,E)$ be a graph with a maximum degree $\Delta<\critD+1$ and $n$ many vertices, and let $\pi$ be a field on $G$.
  Let $\lambda=\max_{v\in V}\{\pi(v)\}$.
  Then $Z_{\beta,\gamma,\pi}(G)$ can be approximated deterministically within a relative error $\epsilon$
  in time $O\left(n\left( \frac{n\lambda}{\epsilon} \right)^{\frac{\log(\Delta-1)}{-\log\alpha}}\right)$,
  where $\alpha=\tfrac{\Delta-1}{\critD}$.
  \label{thm:tractable:first-order}
\end{theorem}

\begin{proof}
  We choose our potential function to be $\Phi_1(x)=\log x$ such that $\phi_1(x):=\Phi_1'(x)=\tfrac{1}{x}$.
  We verify the conditions of Definition \ref{def:potential}.
  Condition \ref{cond:boundedness} is trivial.
  For Condition \ref{cond:contraction}, we have that for any integer $1\le d \le \Delta-1$,
  \begin{align*}
    \ctn{\varphi_1}{d}(\mathbf{x}) & = \phi_1(F_d(\mathbf{x}))\sum_{i=1}^d\frac{\partial F_d}{\partial x_i}\cdot\frac{1}{\phi_1(x)}\\
    &=\frac{1}{F_d(\mathbf{x})}\sum_{i=1}^d F_d(\mathbf{x})\cdot\frac{\beta\gamma-1}{(x_i+\beta)(\gamma x_i+1)}\cdot x_i\\
    &=\sum_{i=1}^d\frac{(\beta\gamma-1)x_i}{(\gamma x_i+1)(x_i+\beta)}\le\sum_{i=1}^d \frac{1}{\critD}=\frac{d}{\critD}\le \frac{\Delta-1}{\critD}=\alpha,
  \end{align*}
  where we used the fact that for any $x>0$,
  \begin{align*}
    \frac{(\beta\gamma-1)x}{(\gamma x+1)(x+\beta)}\le\frac{1}{\critD}.
  \end{align*}
  Hence $\Phi_1(x)$ is a good potential function for all degrees $d\in[1,\Delta-1]$ with contraction ratio $\alpha$.
  The theorem follows by Lemma \ref{lem:potential:algorithm}.
\end{proof}

Note that Theorem \ref{thm:tractable:first-order} matches the uniqueness condition in Proposition \ref{prop:uniqueness:bounded}
and, restricted to the Ising model, the fast mixing bound of Gibbs samplers in \cite{MS13}.

\subsection{General Graphs} \label{sec:algorithm:general}

Now we turn our attention to general graphs without degree bounds.

Recall that $\critL=\left(\frac{\gamma}{\beta}\right)^{\frac{\critD+1}{2}}=\left(\frac{\gamma}{\beta}\right)^{\frac{\sqrt{\beta \gamma}}{\sqrt{\beta \gamma}-1}}$.
The following two technical lemmas show some important properties regarding the threshold $\critL$,
which are keys to get our main theorems.
In particular, Lemma~\ref{lem:inequality-key} is key to bound the decay ratio.
Proofs are given in Section \ref{sec:lemma-proof}.

\begin{lemma}\label{lem:critL}
  Let $\beta,\gamma$ be two parameters such that $\beta\gamma>1$ and $\beta\le\gamma$.
  For any $0<x\le\critL$, $\frac{\beta x+1}{x+\gamma}\leq 1$.
\end{lemma}

\begin{lemma} \label{lem:inequality-key}
  Let $\beta,\gamma$ be two parameters such that $\beta\gamma>1$ and $\beta\le\gamma$.
  For any $0<x\le\critL$, we have
  \begin{align} \label{eqn:inequality-key}
    (\beta \gamma-1)x \log{\frac{\critL}{x}}\le(\beta x+1)(x+\gamma) \log{\frac{x+\gamma}{\beta x+1}}.
  \end{align}
\end{lemma}

In our applications, the quantity $x$ in both lemmas will be the ratio of marginal probabilities in trees, denoted by $R_v$ for a vertex $v$.
To make use of these properties, one key requirement is that $0<x\le\critL$.
This is not necessarily true in trees with pinning (and therefore not true in general SAW trees).
Nevertheless, it does hold in trees without pinning.

\begin{claim}\label{claim:bound}
  For \bgl where $\beta\gamma>1$, $\beta\le\gamma$, and $\lambda<\critL$,
  $R_v\in(0,\lambda]$ holds in trees without pinning.
\end{claim}

We prove Claim \ref{claim:bound} by induction.
For any tree $T_v$, if $v$ is the only vertex, then $R_v=\lambda$ and the base case holds.
Given Lemma \ref{lem:critL} and $\lambda<\critL$, the inductive step to show Claim \ref{claim:bound} follows
from the standard tree recursion \eqref{eqn:tree-recursion}.

In addition, it also holds when $\beta\le 1$, in trees even with pinning (but not counting the pinned vertices).
This includes the SAW tree construction as special cases.
To see that, for any vertex $v$, if one of $v$'s child, say $u$, is pinned to $0$ (or $1$),
then we can just remove $u$ and replace the field $\lambda_v$ on $v$ with $\lambda_v'=\lambda_v\beta$ (or $\lambda_v'=\lambda_v/\gamma$),
without affecting the marginal probability of $v$ and any other vertices.\footnote{We may need to repeat this step for $d$ times, giving rise to the interval in Definition~\ref{def:potential}.}
By our assumptions $\lambda_v<\critL$ and $\beta\le 1<\gamma$, we have that $\lambda_v'<\critL$ as well.
Hence, after removing all pinned vertices, we still have that $\lambda_v\le\critL$ for all $v\in V$.
This reduces to Claim \ref{claim:bound}.

Indeed, both of Theorem \ref{thm:mixing} and \ref{thm:algorithm} can be generalized to the setting
where vertices may have different external fields as long as they are all below $\critL$, as follows.

\begin{theorem} \label{thm:mixing-full}
  Let $(\beta,\gamma)$ be two parameters such that $\beta\gamma>1$, $\beta\le\gamma$, and $\lambda<\critL$.
  Let $T_v$ and $T'_{v'}$ be two trees with roots $v$ and $v'$ respectively.
  Let $\lambda=\max_{u\in T_v \cup T'_{v'}}\{\pi(u)\}$.
  If $\lambda<\critL$ and in the first $\ell$ levels,
  $T_v$ and $T'_{v'}$ have the same structure and external fields for corresponding pairs of vertices,
  then $|p_v-p_{v'}|\leq O(\exp(-\ell))$.
\end{theorem}

\begin{theorem}
  Let $(\beta,\gamma)$ be two parameters such that $\beta\gamma>1$ and $\beta\le 1<\gamma$.
  Let $G=(V,E)$ be a graph with $n$ many vertices, and let $\pi$ be a field on $G$.
  Let $\lambda=\max_{v\in V}\{\pi(v)\}$.
  If $\lambda<\critL$, then $Z_{\beta,\gamma,\pi}(G)$ can be approximated deterministically within a relative error $\epsilon$
  in time $O\left(n\left( \frac{n\lambda}{\epsilon} \right)^{\frac{\log M}{-\log\alpha}}\right)$,
  where $M>1$ and $\alpha<1$ are two constants depending on $(\beta,\gamma,\lambda)$.
  \label{thm:tractable:beta<1}
\end{theorem}

To show Theorem \ref{thm:mixing-full} and Theorem \ref{thm:tractable:beta<1},
we will apply Lemma \ref{lem:potential:spatial-mixing} and Lemma \ref{lem:potential:algorithm:universal}.
Essentially we only need to show the existence of a universal potential function.

Let $g_\lambda(x):=\frac{(\beta \gamma-1)x \log{\frac{\lambda}{x}}}{(\beta x+1)(x+\gamma) \log{\frac{x+\gamma}{\beta x+1}}}$.
By Lemma \ref{lem:inequality-key}, $g_{\lambda_c}(x)\le 1$. For $\lambda<\lambda_c$,
note that $\lim_{x\rightarrow 0}g_\lambda(x)=0$.
Hence there exists $0<\epsilon<\lambda$ and $0<\delta<1$ such that if $0<x<\epsilon$, $g_\lambda(x)<\delta$.
Moreover, if $\epsilon\le x\le \lambda$,
then $\frac{g_\lambda(x)}{g_{\critL}(x)}=\frac{\log\lambda-\log x}{\log\critL-\log x}\le\frac{\log\lambda-\log \epsilon}{\log\critL-\log \epsilon}$.
Let
\begin{align*}
  \alpha_\lambda:=\max\left\{\delta,\frac{\log\lambda-\log \epsilon}{\log\critL-\log \epsilon}\right\}<1.
\end{align*}
Then we have just shown the following lemma.

\begin{lemma}
  Let $\beta,\gamma$ be two parameters such that $\beta\gamma>1$ and $\beta\le\gamma$.
  If $\lambda<\critL$, then $g_\lambda(x) \leq \alpha_\lambda$ for any $0<x\le\lambda$,
  where $\alpha_\lambda<1$ is defined above.
  \label{lem:alpha-existence}
\end{lemma}

Let $t:=\frac{\alpha_\lambda\gamma}{\beta \gamma-1} \log{\frac{\lambda+\gamma}{\beta \lambda+1}}$
so that for any $0<x\le\lambda$,
\begin{align}\label{eqn:t}
  t < \frac{\alpha_\lambda(\beta x+1)(x+\gamma)}{\beta \gamma-1} \log{\frac{x+\gamma}{\beta x+1}},
\end{align}
since $(\beta x+1)(x+\gamma)>\gamma$ and $\log{\frac{x+\gamma}{\beta x+1}}\ge \log{\frac{\lambda+\gamma}{\beta \lambda+1}}$.
We define $\phi_2(x):=\min\left\{\frac{1}{t},\frac{1}{x\log{\frac{\lambda}{x}}}\right\}$.
To be more specific, note that $x\log\frac{\lambda}{x}\le\frac{\lambda}{e}$ for any $0<x\le\lambda$.
If $t\ge \frac{\lambda}{e}$, then $\frac{1}{x\log{\frac{\lambda}{x}}}\ge \frac{1}{t}$ for any $0<x\le\lambda$, implying that
\begin{align}
  \phi_2(x) = \frac{1}{t}.
  \label{eqn:phi2:definition1}
\end{align}

Otherwise $t<\frac{\lambda}{e}$,
and there are two roots to $x\log{\frac{\lambda}{x}}=t$ in $(0,\lambda]$.
Denote them by $x_0$ and $x_1$.
Then we have that
\begin{eqnarray}
  \phi_2(x) =
  \begin{cases}
    \frac{1}{t}                        & 0  \le x < x_0; \\
    \frac{1}{x\log{\frac{\lambda}{x}}} & x_0\le x < x_1; \\
    \frac{1}{t}                        & x_1\le x < \lambda .
  \end{cases}
  \label{eqn:phi2:definition2}
\end{eqnarray}
We define $\Phi_2(x):=\int_{0}^x\phi_2(y)dy$ so that $\Phi_2'(x)=\phi_2(x)$.
Our choice of $\phi_2(x)$ ensures that for any $0<x\le\lambda$,
\begin{align}
  \phi_2(x) x\log{\frac{\lambda}{x}}\le 1.
  \label{eqn:phi2-bounded}
\end{align}
Moreover, we claim that
\begin{align}
  \frac{\beta \gamma-1}{(\beta x+1)(x+\gamma)} \cdot\frac{1}{\phi_2(x)} \leq \alpha_\lambda \log{\frac{x+\gamma}{\beta x+1}}.
  \label{eqn:phi2-ineq}
\end{align}
This is because if $\phi_2(x)=\frac{1}{x\log{\frac{\lambda}{x}}}$, then by Lemma \ref{lem:critL} and Lemma \ref{lem:alpha-existence},
\begin{align*}
  \frac{\beta \gamma-1}{(\beta x+1)(x+\gamma)} \cdot x\log{\frac{\lambda}{x}} \leq \alpha_\lambda \log{\frac{x+\gamma}{\beta x+1}},
\end{align*}
which implies \eqref{eqn:phi2-ineq}.
Otherwise $\phi_2(x)=1/t$, and \eqref{eqn:phi2-ineq} follows from \eqref{eqn:t}.

Now, we are ready to prove Theorems \ref{thm:mixing-full} and \ref{thm:tractable:beta<1}.

\begin{proof}[Proof of Theorems \ref{thm:mixing-full} and \ref{thm:tractable:beta<1}]
  We claim that $\Phi_2(x)$ is a universal potential function for any field $\pi$ with an upper bound $\lambda$,
  with contraction ratio $\alpha_\lambda$ given above and base $M$ that will be determined shortly.
  Theorem \ref{thm:mixing-full} and Theorem \ref{thm:tractable:beta<1} follow from $\Phi_2(x)$
  combined with Lemma \ref{lem:potential:spatial-mixing} and \ref{lem:potential:algorithm:universal}, respectively.
  We verify the two conditions in Definition \ref{def:potential:universal}.

  For Condition \ref{cond:universal-boundedness},
  it is easy to see that in case \eqref{eqn:phi2:definition1}, $\phi_2(x)=\frac{1}{t}$ for any $x\in(0,\lambda]$,
  and in case \eqref{eqn:phi2:definition2}, $\frac{e}{\lambda}\le \phi_2(x)\le \frac{1}{t}$ for any $x\in(0,\lambda]$.

  For Condition \ref{cond:universal-contraction},
  we have that
  \begin{align}
    \ctn{\phi_2}{d}(\mathbf{x}) & = \phi_2(F_d(\mathbf{x}))\sum_{i=1}^d\frac{\partial F_d}{\partial x_i}\cdot\frac{1}{\phi_2(x_i)} \notag\\
    & = \phi_2(F_d(\mathbf{x}))F_d(\mathbf{x})\sum_{i=1}^d\frac{\beta\gamma-1}{(\beta x_i+1)( x_i+\gamma)}\cdot\frac{1}{\phi_2(x_i)}\notag\\
    & \le \phi_2(F_d(\mathbf{x}))F_d(\mathbf{x}) \sum_{i=1}^d\alpha_\lambda \log{\frac{x_i+\gamma}{\beta x_i+1}} \tag*{(by \eqref{eqn:phi2-ineq})}\\
    & = \alpha_\lambda \phi_2(F_d(\mathbf{x}))F_d(\mathbf{x}) \log\frac{\lambda}{F_d(\mathbf{x})} \label{eqn:ctnphi_2}\\
    & \le \alpha_\lambda. \hspace*{\fill} \tag*{(by \eqref{eqn:phi2-bounded})}
  \end{align}
  Moreover, $F_d(\mathbf{x})< \lambda\left(\frac{\beta\lambda+1}{\lambda+\gamma}\right)^d$ for any $x_i\in(0,\lambda]$.
  We have an alternative bound that
  \begin{align*}
    \ctn{\phi_2}{d}(\mathbf{x}) & \le \frac{\alpha_\lambda}{t} F_d(\mathbf{x}) \log\frac{\lambda}{F_d(\mathbf{x})} \tag*{(by \eqref{eqn:ctnphi_2} and $\varphi_2(x)\le 1/t$)}\\
    & \le \frac{\alpha_\lambda\lambda}{t}\left(\frac{\beta\lambda+1}{\lambda+\gamma}\right)^d d\log\frac{\lambda+\gamma}{\beta\lambda+1}.
  \end{align*}
  Since $\frac{\beta\lambda+1}{\lambda+\gamma}<1$ by Lemma \ref{lem:critL}, the right hand side decreases exponentially in $d$.
  Therefore, there exists a sufficiently large integer $M$ 
  such that for any $1\le d < M$, $\ctn{\phi_2}{d}(\mathbf{x})\le\alpha_\lambda\le\alpha_\lambda^{\ceil{\log_M(d+1)}}$,
  and for any $d\ge M$, $\ctn{\phi_2}{d}(\mathbf{x})
  \le \alpha_\lambda^{\ceil{\log_M(d+1)}}$.
  This verifies Condition \ref{cond:universal-contraction}.
\end{proof}

\subsection{Heuristics behind \texorpdfstring{$\Phi_2(x)$}{Phi2(x)}}

The most intricate part of our proofs of Theorem \ref{thm:mixing-full} and Theorem \ref{thm:tractable:beta<1}
is the choice of the potential function $\Phi_2(x)$ given by \eqref{eqn:phi2:definition2}.
Here we give a brief heuristic of deriving it.
It is more of an ``educated guess'' than a rigorous argument.

We want to pick $\Phi_2(x)$ such that Condition \ref{cond:universal-contraction} holds.
In particular, we want
\[\phi_2(F_d(\mathbf{x}))\sum_{i=1}^d\frac{\partial F_d}{\partial x_i}\cdot\frac{1}{\phi_2(x_i)}<1.\]
It is fair to assume that the left hand side of the equation above takes its maximum when all $x_i$'s are equal.
Hence, we hope the following to hold
\begin{align}
  \frac{\phi_2(f_d(x))f_d'(x)}{\phi_2(x)}<1,
  \label{eqn:heuristic-symmetric}
\end{align}
where $f_d(x)=\lambda\left(\frac{\beta x+1}{x+\gamma}\right)^d$ is the symmetrized version of $F_d(\mathbf{x})$.
We will use $z:=f_d(x)$ to simplify notation.
Since we want \eqref{eqn:heuristic-symmetric} to hold for all degrees $d$,
we hope to eliminate $d$ from the left hand side of \eqref{eqn:heuristic-symmetric}.
Notice that $\phi_2(x)$ should be independent from $d$.
Therefore, we take the derivative of $\phi_2(f_d(x))f_d'(x)$ against $d$ and get
\begin{align*}
  \frac{\partial \phi_2(f_d(x))f_d'(x)}{\partial d} = & \frac{\beta\gamma-1}{(\beta x+1)(x+\gamma)}
  \left(\phi_2(z)z + \phi_2(z) z \log\frac{z}{\lambda}  + \phi_2'(z) z^2 \log\frac{z}{\lambda}\right)\\
  = & \frac{(\beta\gamma-1)z\phi_2(z)}{(\beta x+1)(x+\gamma)}
  \left(1 + \log\frac{z}{\lambda}  + (\log\phi_2(z))' z \log\frac{z}{\lambda}\right).
\end{align*}
We may achieve our goal of eliminating $d$ by imposing the sum in the last parenthesis to be $0$, namely
\begin{align}
  (\log\phi_2(z))' & = - \frac{1}{z} - \frac{1}{z \log\frac{z}{\lambda}}\notag\\
  & = -(\log z)' - \left(\log\log \frac{\lambda}{z}\right)'.
  \label{eqn:heuristic-differential-equation}
\end{align}
From \eqref{eqn:heuristic-differential-equation},
it is easy to see that $\phi_2(z)=\frac{1}{z\log\frac{\lambda}{z}}$ satisfies our need.
To get the full definition of \eqref{eqn:phi2:definition2},
we apply a thresholding trick to bound $\phi_2(z)$ away from $0$.

\subsection{Discussion of the \texorpdfstring{$\beta>1$}{beta>1} case}

We cannot combine conditions of Theorem \ref{thm:mixing-full} and Theorem \ref{thm:tractable:beta<1} together to have an FPTAS.
In particular, when $\beta>1$ strong spatial mixing fails for any $\lambda$ even if $\lambda<\critL$.
To see this, given a $\Delta$-ary tree $T$,
we can append $t$ many children to every vertex in $T$ to get a new tree $T'$
and impose a partial configuration $\sigma$ where all these new children are pinned to $0$.
Effectively, the tree $T'$ is equivalent to $T$ where every vertex has a new external field of $\lambda\beta^t$,
which is larger than $\critLint$ if $t$ is sufficiently large regardless of $\lambda$.
Then by Proposition \ref{prop:uniqueness:general}, long range correlation exists in $T'$ with the partial configuration $\sigma$,
and strong spatial mixing fails.

On the other hand, it is easy to see from the proof that,
Theorem \ref{thm:mixing-full} can be generalized to allow a partial configuration $\sigma$ on some subset $\Lambda$
where the marginal probability of every vertex $v\in\Lambda$ satisfies $p^\sigma_v\le\frac{\critL}{\critL+1}$.
This is not the case for the SAW tree which our algorithm relies on when $\beta>1$.
However, the following observation shows that if $\lambda_v\le\critL\le\frac{\gamma-1}{\beta-1}$,
then the marginal probability of any instance $G$ satisfies this requirement.
Thus, it seems the only piece missing to obtain an algorithm is to design a better recursion tree instead of the SAW tree.

\begin{proposition}\label{prop:bounded}
  Let $(\beta,\gamma)$ be two parameters such that $1\le\beta\le\gamma$ and $\beta\gamma>1$.
  Let $\lambda\le\frac{\gamma-1}{\beta-1}$ be another parameter.
  For any graph $G=(V,E)$, if $\pi(v)\le\lambda$ for all $v\in V$,
  then $p_{v}\le\frac{\lambda}{\lambda+1}$.
\end{proposition}

To prove this proposition, we need to use the random cluster formulation of $2$-spin models.
Let $G$ be a graph and $e=(v_1,v_2)$ be one of its edges.
Let $G^+$ be the graph where the edge $e$ is contracted,
and $G^-$ be the graph where $e$ is removed.
Moreover, in $G^+$, we assign $\pi^+(\widetilde{v})=\lambda_{v_1}\lambda_{v_2}\frac{\beta-1}{\gamma-1}$,
where $\widetilde{v}$ is the vertex obtained from contacting $e$.
Then we have that
\begin{align}
  \label{eqn:random-cluster}
  Z(G) = Z(G^-)+(\gamma-1)Z(G^+),
\end{align}
where we write $Z(G)$ instead of $Z_{\beta,\gamma,\pi}(G)$ to simplify the notation.
To show the equation above we only need a simple adapation of the random cluster formulation of the Ising model to the $2$-spin setting.

\begin{proof}[Proof of Proposition \ref{prop:bounded}]
  Suppose $G=(V,E)$ where $|V|=n$ and $|E|=m$.
  We show the claim by inducting on $(m,n)$.
  Clearly the statement holds when $m=0$ or $n=1$.
  Hence we may assume the claim holds for $(m',n)$ where $m'< m$ as well as $(m',n')$ where $n'< n$,
  and show that the claim holds for $(m,n)$.

  Pick an arbitrary edge $e=(v_1,v_2)$ in $G$.
  Let $G^+$ and $G^-$ be as in the random cluster formulation.
  It is easy to see that $\pi(\widetilde{v})=\lambda_{v_1}\lambda_{v_2}\frac{\beta-1}{\gamma-1}\le\lambda$.
  Hence both $G^+$ and $G^-$ satisfy the induction hypothesis.
  It implies that $p_{G^-;v}\le\frac{\lambda}{\lambda+1}$ for any $v$,
  where $p_{G^-;v}$ is the mariginal probability of $v$ in $G^-$.
  Moreover, $p_{G^+;v}\le\frac{\lambda}{\lambda+1}$ for any $v\in V^+$,
  where $V^+$ is the vertex set of $G^+$.
  Let $\delta$ be a mapping $V\rightarrow V^+$ such that $\delta(v)=v$ if $v\neq v_1,v_2$ and $\delta(v_1)=\delta(v_2)=\widetilde{v}$.
  Then using \eqref{eqn:random-cluster} we have that for any vertex $v\in V$,
  \begin{align*}
    p_{G;v} & = \frac{Z^{\sigma(v)=0}(G)}{Z(G)}
    = \frac{Z^{\sigma(v)=0}(G^-)+(\gamma-1)Z^{\sigma(\delta(v))=0}(G^+)}
    {Z(G^-)+(\gamma-1)Z(G^+)}\\
    & = p_{G^-;v} \cdot\frac{Z(G^-)}{Z(G^-)+(\gamma-1)Z(G^+)}
    + p_{G^+;\delta(v)} \cdot \frac{(\gamma-1)Z(G^+)}{Z(G^-)+(\gamma-1)Z(G^+)}\\
    & \le \frac{\lambda}{\lambda+1}\cdot \frac{Z(G^-)}{Z(G^-)+(\gamma-1)Z(G^+)}
    + \frac{\lambda}{\lambda+1}\cdot \frac{(\gamma-1)Z(G^+)}{Z(G^-)+(\gamma-1)Z(G^+)} = \frac{\lambda}{\lambda+1},
  \end{align*}
  where in the last line we use the induction hypotheses.
\end{proof}

Proposition \ref{prop:bounded} can be also viewed as a generalization of Griffith's first inequality~\cite{Gri72} from the Ising model to general ferromagnetic 2-spin systems.

\section{Correlation Decay Beyond \texorpdfstring{$\critL$}{the Critical Field}} \label{sec:beyond-critL}

Let $\beta,\gamma$ be two parameters such that $\beta\le 1 < \gamma$ and $\beta\gamma>1$.
In this section we give an example to show that if $\critD$ is not an integer, then correlation decay still holds for a small interval beyond $\critL$.
To simplify the presentation, we assume that $\pi$ is a uniform field such that $\pi(v)=\lambda$.
Note that the potential function $\phi_2(x)$ does not extend beyond $\critL$.

Let $\beta=0.6$ and $\gamma=2$.
Then $\critD=\frac{\sqrt{\beta\gamma}+1}{\sqrt{\beta\gamma}-1}\approx 21.95$
and $\critL=\left(\gamma/\beta\right)^{\frac{\critD+1}{2}} < 1002761$.
Let $\lambda=1002762>\critL$.
We will show that \Spin{\beta}{\gamma}{\lambda} still has an FPTAS.

Define a constant $t$ as
\begin{align}\label{eqn:def:t}
  t:=
  \frac{\sqrt{\beta\gamma} + 1}{\sqrt{\beta\gamma}-1} \cdot \frac{\log\sqrt{\gamma/\beta}} {\sqrt{\gamma/\beta} + 1}
  -\log\left(1 + \sqrt{\beta/\gamma}\right)
  \approx 4.24032.
\end{align}
We consider a potential function $\Phi_3(x)$ so that $\phi_3(x):=\frac{1}{x(\log(1+1/x)+t)}$.
With this choice,
\begin{align}
  \ctn{\phi_3}{d}(\mathbf{x})&=\phi_3(F_d(\mathbf{x}))\sum_{i=1}^d\frac{\partial F_d}{\partial x_i}\cdot\frac{1}{\phi_3(x)}\notag\\
  &=\frac{\beta\gamma-1}{\log\left(1+1/F_d(\mathbf{x})\right)+t}
  \sum_{i=1}^d\frac{x_i\left(\log(1+1/x_i)+t\right)}{(\beta x_i+1)(x_i+\gamma)}.\notag
\end{align}

We do a change of variables.
Let $r_i=\frac{\beta x_i+1}{x_i+\gamma}$.
Then $x_i=\frac{\gamma r_i-1}{\beta-r_i}$,
$\beta x_i+1=\frac{r_i(\beta\gamma-1)}{\beta-r_i}$,
and $x_i+\gamma=\frac{\beta\gamma-1}{\beta-r_i}$.
Hence,
\begin{align*}
  \sum_{i=1}^d\frac{x_i(\log(1+1/x_i)+t)}{(\beta x_i+1)(x_i+\gamma)}
  &=\sum_{i=1}^d\frac{(\gamma r_i-1)(\beta-r_i)}{r_i(\beta\gamma-1)^2}\cdot \left(\log\left(1+\frac{\beta-r_i}{\gamma r_i-1}\right)+t\right)\\
  &=\frac{1}{(\beta\gamma-1)^2}\sum_{i=1}^d\left( 1+\beta\gamma-\frac{\beta}{r_i}-\gamma r_i \right) \left(\log\left(1+\frac{\beta-r_i}{\gamma r_i-1}\right)+t\right).
\end{align*}
Furthermore, let $s_i=\log r_i$.
As $r_i\in\left(\frac{1}{\gamma},\beta\right)$, $s_i\in \left( -\log\gamma,\log\beta \right)$.
Let
\begin{align*}
  \rho(x):=\left( 1+\beta\gamma-\beta e^{-x}-\gamma e^x \right) \left(\log\left(1+\frac{\beta-e^x}{\gamma e^x-1}\right)+t\right).
\end{align*}
Then $\rho(x)$ is concave for any $x\in \left( -\log\gamma,\log\beta \right)$.
It can be easily verified, as the second derivative is
\begin{align} \label{eqn:concavity}
  \rho''(x)= \;& 
  \frac{(\beta+1)(\beta\gamma-1)}{\beta-1+e^x(\gamma-1)}+\frac{\beta\gamma-1}{\gamma-1}-\frac{\beta\gamma-1}{e^x\gamma-1}\notag
  - \frac{(\beta-1) (\beta\gamma-1)^2}{(\gamma-1)(\beta-1 + e^x (\gamma-1))^2 }\\
  &- \beta t e^{-x} - \gamma t e^x - e^{-x} \left(\beta + e^{2 x} \gamma\right) \log\left(1 + \frac{\beta - e^x}{\gamma e^x-1}\right).\notag\\
  \le\; & \gamma(\beta+1) + \frac{\beta\gamma-1}{\gamma-1} - 1-\frac{\beta-1}{\gamma-1}-2t < -5<0, 
\end{align}
where in the last line we used \eqref{eqn:def:t} and the fact that $1/\gamma\le e^x\le \beta$.
Hence, by concavity, we have that for any $x_i\in(0,\lambda]$,
\begin{align}
  \ctn{\phi_3}{d}(\mathbf{x})&=
  \frac{\beta\gamma-1}{\log\left(1+1/F_d(\mathbf{x})\right)+t}\sum_{i=1}^d\frac{x_i\left(\log(1+1/x_i)+t\right)}{(\beta x_i+1)(x_i+\gamma)}, \notag \\
  &\le \frac{\beta\gamma-1}{\log\left(1+1/f_d(\widetilde{x})\right)+t}\cdot
  \frac{d \widetilde{x}\left(\log(1+\widetilde{x}^{-1})+t\right)}{(\beta \widetilde{x}+1)(\widetilde{x}+\gamma)}=\symctn{\phi_3}{d}(\widetilde{x}),
  \label{eqn:symmetrization}
\end{align}
where $\widetilde{x}>0$ is the unique solution such that $f_d(\widetilde{x})=F_d(\mathbf{x})$.

Next we show that there exists an $\alpha<1$ such that for any integer $d$ and $x>0$, $\symctn{\phi_3}{d}(x)<\alpha$.
In fact, by \eqref{eqn:def:t}, our choice of $t$,
it is not hard to show that the maximum of $\symctn{\phi_3}{d}(x)$ is achieved at $x=\sqrt{\gamma/\beta}$ and $d=\critD$,
which is $1$ if $\lambda=\critL$ and is larger than $1$ if $\lambda>\critL$.
However, since the degree $d$ has to be an integer, we can verify that for any integer $1\le d\le 100$,
the maximum of $\symctn{\phi_3}{d}(x)$ is $\symctn{\phi_3}{22}(x_{22})=0.999983$ where $x_{22}\approx 1.83066$.
If $d>100$, then
\begin{align*}
  \symctn{\phi_3}{d}(x) & = \frac{d(\beta\gamma-1)}{\log\left(1+1/f_d(x)\right)+t}\cdot
  \frac{x\left(\log(1+x^{-1})+t\right)}{(\beta x+1)(x+\gamma)}\\
  & \le C_0 \cdot C_1 <1,
\end{align*}
where $C_0 < 1.07191$ is the maximum of $\frac{x\left(\log(1+x^{-1})+t\right)}{(\beta x+1)(x+\gamma)}$ for any $x>0$,
and $C_1 < 0.481875$ is the maximum of $\frac{d(\beta\gamma-1)}{\log\left(1+\lambda^{-1}\beta^{-d}\right)+t}$ for any $d>100$.
Then, due to \eqref{eqn:symmetrization}, we have that for any $x_i\in(0,\lambda]$, $\ctn{\phi_3}{d}(\mathbf{x})<\alpha=0.999983<1$.
This is the counterpart of $\ctn{\phi_2}{d}(\mathbf{x})<\alpha_\lambda$ in the proof of Theorem \ref{thm:tractable:beta<1}.
To make $\phi_3(x)$ satisfy Condition \ref{cond:universal-boundedness} and Condition \ref{cond:universal-contraction} in Definition \ref{def:potential:universal},
it is sufficient to do a simple ``chop-off'' trick to $\phi_3(x)$ as in \eqref{eqn:phi2:definition2}.
We will omit the detail here.

\begin{proposition} \label{prop:beyond-critL}
  For $\beta=0.6$, $\gamma=2$, and $\lambda=1002762>\critL$, \Spin{\beta}{\gamma}{\lambda} has an FPTAS.
\end{proposition}

It is easy to see that the argument above works for any $\beta\le 1<\gamma$ and $\beta\gamma>1$ except \eqref{eqn:concavity}, the concavity of $\rho(x)$.
Indeed, the concavity does not hold if, say, $\beta=1$ and $\gamma=2$.
Nevertheless, the key point here is that $\critL$ is not the tight bound for FPTAS.
Short of a conjectured optimal bound, we did not try to optimize the potential function nor the applicable range of the proof above.

\section{Limitations of Correlation Decay} \label{sec:hardness}

In this section, we discuss some limitations of approximation algorithms for ferromagnetic 2-spin models based on correlation decay analysis.

The problem of counting independent sets in bipartite graphs (\BIS) plays an important role in classifying approximate counting complexity.
\BIS\ is not known to have any efficient approximation algorithm, despite many attempts.
However there is no known approximation preserving reduction (AP-reduction) to reduce \BIS\ from \SAT\ either.
It is conjectured to have intermediate approximation complexity, and in particular, to have no FPRAS \cite{DGGJ03}.

Goldberg and Jerrum \cite{GJ07} showed that for any $\beta\gamma>1$, approximating \Spin{\beta}{\gamma}{(0,\infty)} can be reduced to approximating \BIS.
This is the (approximation) complexity upper bound of all ferromagnetic 2-spin models.
In contrast, by Theorem \ref{thm:tractable:first-order}, \dSpin{\beta}{\gamma}{(0,\infty)}{\Delta} has an FPTAS, if $\Delta<\critD+1$.
Note that when we write \Spin{\beta}{\gamma}{(0,\infty)} the field is implicitly assumed to be at most polynomial in size of the graph (or in unary).

We then consider fields with some constant bounds.
Recall that $\critLint=(\gamma/\beta)^{\frac{\ceil{\critD}+1}{2}}$.
Let $\critLint'=(\gamma/\beta)^{\frac{\floor{\critD}+2}{2}}$.
Then $\critLint'=\critLint$ unless $\critD$ is an integer.
By reducing to anti-ferromagnetic 2-spin models in bipartite graphs, we have the following hardness result,
which is first observed in \cite[Theorem 3]{LLZ14a}.

\begin{proposition}\label{prop:hardness}
  Let $(\beta,\gamma,\lambda)$ be a set of parameters such that $\beta<\gamma$, $\beta\gamma>1$, and $\lambda>\critLint'$.
  Then \Spin{\beta}{\gamma}{(0,\lambda]} is \BIS-hard.
\end{proposition}

The reduction goes as follows.
Anti-ferromagnetic Ising models with a constant non-trivial field in bounded degree bipartite graphs are \BIS-hard, if the uniqueness condition fails \cite{CGGGJSV16}.
Given such an instance, we may first flip the truth table of one side.
This effectively results in a ferromagnetic Ising model in the same bipartite graph, with two different fields on each side.
By a standard diagonal transformation, we can transform such an Ising model to any ferromagnetic 2-spin model, with various local fields depending on the degree.
It can be verified that for any $\lambda>\critLint'$, we may pick a field in the anti-ferromagnetic Ising model to start with,
such that uniqueness fails and after the transformation, the largest field in use is at most $\lambda$.

The hardness bound in Proposition \ref{prop:hardness} matches the failure of uniqueness due to Proposition \ref{prop:uniqueness:general},
unless $\critD$ is an integer.
In contrast to Proposition \ref{prop:hardness}, Theorem \ref{thm:tractable:beta<1} implies that
if $\beta\le 1<\gamma$ and $\lambda<\critL=(\gamma/\beta)^{\frac{\critD+1}{2}}$, then \Spin{\beta}{\gamma}{(0,\lambda]} has an FPTAS.
Hence Theorem \ref{thm:tractable:beta<1} is almost optimal, up to an integrality gap.

We note that $\critL$ is not the tight bound for FPTAS, as observed in Proposition \ref{prop:beyond-critL}.
Since the degree $d$ has to be an integer, with an appropriate choice of the potential function,
there is a small interval beyond $\critL$ such that strong spatial mixing still holds.
Interestingly, it seems that $\critLint$ is not the right bound either.
Let us make a concrete example.
Let $\beta=1$ and $\gamma=2$.
Then $\critD=\frac{\sqrt{\beta\gamma}+1}{\sqrt{\beta\gamma}-1}=\frac{\sqrt{2}+1}{\sqrt{2}-1} \approx 5.82843$.
Hence $\critL\approx 10.6606$ and $\critLint = (2)^{\frac{6+1}{2}}  \approx 11.3137 $.
However, even if $\lambda<\critLint$, the system may not exhibit spatial mixing, neither in the strong nor in the weak sense.

In fact, even the spatial mixing in the sense of Theorem \ref{thm:mixing} does not necessarily hold if $\lambda<\critLint$.
To see this, we take any $\lambda\in[10.9759,10.9965]$ so that $\critL<\lambda< \critLint$.
Consider an infinite tree where at even layers, each vertex has $5$ children, and at odd layers, each vertex has $7$ children.
There are more than one Gibbs measures in this tree.
This can be easily verified from the fact that the two layer recursion function $f_5(f_7(x))$ has three fixed points such that $x=f_5(f_7(x))$.
In addition, all three fixed points $\widehat{x}_i$ satisfy that $\widehat{x}_i<\critL$ for $i=1,2,3$.
Consider a tree $T$ with alternating degrees $6$ and $8$ of depth $2\ell$ (so that the number of children is alternatingly $5$ and $7$),
and another tree $T'$ of the same structure in the first $2\ell$ layers as $T$
but with one more layer where each vertex has, say, $50$ children.
It is not hard to verify that as $\ell$ increases, the marginal ratio at the root of $T$ converges to $\widehat{x}_3$,
but the ratio at the root of $T'$ converges to $\widehat{x}_1$.
This example indicates that one should not expect correlation decay algorithms to work all the way up to $\critLint$.

At last, if we consider the uniform field case \Spin{\beta}{\gamma}{\lambda}, then our tractability results still holds.
However, to extend the hardness results as in Proposition \ref{prop:hardness} from an interval of fields to a uniform one,
there seems to be some technical difficulty.
Suppose we want to construct a combinatorial gadget to effectively realize another field.
There is a gap between $\lambda$ and the next largest possible field to realize.
This is why in \cite{LLZ14a}, there are some extra conditions transiting from an interval of fields to the uniform case.
The observation above about the failure of SSM in irregular trees may suggest a random bipartite construction of uneven degrees.
However, to analyze such a gadget is beyond the scope of the current paper.

\section{Missing Proofs}

At last, we gather technical details and proofs that are omitted in Section~\ref{sec:uniqueness}, Section~\ref{sec:potential}, and Section~\ref{sec:algorithm:general}.

\subsection{Details about the Uniqueness Threshold}\label{sec:missing-proofs:uniqueness}

We prove Propositions \ref{prop:uniqueness:bounded} and Proposition \ref{prop:uniqueness:general}.
Technically by only considering the symmetric recursion $f_d(x)=\lambda\left( \frac{\beta x+1}{x+\gamma} \right)^d$,
we are implicitly assuming uniform boundary conditions.
If there are more than one fixed points for $f_d(x)$, then clearly there are multiple Gibbs measures.
Hence, $f_d(x)$ having only one fixed point is a necessary condition for the uniqueness condition in \tree{d+1}.
Moreover, it is also sufficient.
The reason is that the influence on the root of an arbitrary boundary condition is bounded
between those of the all ``0'' and all ``1'' boundary conditions.

First do some calculation here.
Take the derivative of $f_d(x)$:
\begin{align} \label{eqn:f:derivative}
  f_d'(x)=\frac{d(\beta\gamma-1)f_d(x)}{(\beta x+1)(x+\gamma)}.
\end{align}
Then take the second derivative:
\begin{align*}
  \frac{f_d''(x)}{f_d'(x)}  &=\frac{f_d'(x)}{f_d(x)} - \frac{\beta}{\beta x+1} - \frac{1}{x+\gamma}=\frac{d(\beta\gamma-1)-\beta\gamma-1-2\beta x}{(\beta x+1)(x+\gamma)}.
\end{align*}
Therefore, at $x^*:=\frac{d(\beta\gamma-1)-(\beta\gamma+1)}{2\beta}$, $f''_d(x^*)=0$.
It's easy to see when $d<\frac{\beta\gamma+1}{\beta\gamma-1}$, $f''_d(x)<0$ for all $x>0$.
So $f_d(x)$ is concave and therefore has only one fixed point.

Since $f_d(x)$ has only one inflection point, there are at most three fixed points.
Moreover, the uniqueness condition is equivalent to say that for all fixed points $\widehat{x}_d$ of $f_d(x)$, $f_d'(\widehat{x}_d)<1$.
For a fixed point $\widehat{x}_d$, we plug it in \eqref{eqn:f:derivative}:
\begin{align*}
  f_d'(\widehat{x}_d)=\frac{d(\beta\gamma-1)\widehat{x}_d}{(\beta \widehat{x}_d+1)(\widehat{x}_d+\gamma)}.
\end{align*}
Recall that $\critD:=\frac{\sqrt{\beta\gamma}+1}{\sqrt{\beta\gamma}-1}$.
If $d<\critD$, we have that for any $x$,
\begin{align*}
  (\beta x+1)(x+\gamma)-d(\beta\gamma-1)x
  &=\beta x^2+((\beta\gamma+1)-d(\beta\gamma-1))x+\gamma\\
  &>\beta x^2+(\beta\gamma+1-(\sqrt{\beta\gamma}+1)^2)x+\gamma\\
  &=(\sqrt{\beta}x-\sqrt{\gamma})^2\geq 0.
\end{align*}
Hence $(\beta x+1)(x+\gamma)>d(\beta\gamma-1)x$.
In particular, $f_d'(\widehat{x}_d)<1$ for any fixed point $\widehat{x}_d$ and the uniqueness condition holds.
This proves Proposition \ref{prop:uniqueness:bounded}.

To show Proposition \ref{prop:uniqueness:general}, we may assume that $d\geq \critD$.
We may also assume that $\beta\le\gamma$.
The equation $(\beta x+1)(\gamma+x)=d(\beta\gamma-1)x$ has two solutions, which are
\begin{align*}
  x_0&=x^*-\frac{\sqrt{((\beta\gamma+1)-d(\beta\gamma-1))^2-4\beta\gamma}}{2\beta}\\
  \text{\quad and \quad}
  x_1&=x^*+\frac{\sqrt{((\beta\gamma+1)-d(\beta\gamma-1))^2-4\beta\gamma}}{2\beta}.
\end{align*}
Notice that both of them are positive since $x_0+x_1=2 x^*>0$ and $x_0x_1=\gamma/\beta$.
As $d$ goes to $\infty$, 
\begin{align}  \label{eqn:x0x1:asymp}
  x_0&=o(1), & x_1&=2x^* - o(1) = \frac{d(\beta\gamma-1)-(\beta\gamma+1)}{\beta} - o(1).
\end{align}
Moreover, 
\begin{align}  \label{eqn:x0x1}
  \frac{d(\beta\gamma-1)x}{(\beta x+1)(\gamma+x)}>1 \text{ \quad if and only if \quad } x_0 < x < x_1.
\end{align}

We show that $f_d(x_0)>x_0$ or $f_d(x_1)<x_1$ is equivalent to the uniqueness condition.
First we assume this condition does not hold,
that is $f_d(x_0)\leq x_0$ and $f_d(x_1)\geq x_1$.
If any of the equation holds,
then $x_0$ or $x_1$ is a fixed point and the derivative is $1$.
So we have non-uniqueness.
Otherwise, we have $f_d(x_0)<x_0$ and $f_d(x_1)>x_1$.
Since $x_0<x_1$, there is some fixed point $\widetilde{x}$ satisfying $f_d(\widetilde{x})=\widetilde{x}$
and $x_0<\widetilde{x}<x_1$.
The second inequality implies that $f_d'(\widetilde{x})>1$ via \eqref{eqn:x0x1} and non-uniqueness holds.

To show the other direction, if $f_d(x_0)>x_0$, then
\begin{align*}
  f_d'(x_0)
  &=\frac{d(\beta\gamma-1)f(x_0)}{(\beta x_0+1)(x_0+\gamma)}>\frac{d(\beta\gamma-1)x_0}{(\beta x_0+1)(x_0+\gamma)} =1.
\end{align*}
Assume for contradiction that $f_d(x)$ has three fixed points, denoted by $\widetilde{x}_0<\widetilde{x}_1<\widetilde{x}_2$.
Then the middle fixed point $\widetilde{x}_1$ satisfies $f_d'(\widetilde{x}_1)>1$.
Therefore $\widetilde{x}_1>x_0$ by \eqref{eqn:x0x1} and there are two fixed points larger than $x_0$.
However, for $x_0<x\leq x^*$, $f_d'(x)>1$ and $f_d(x_0)>x_0$. Hence there is no fixed point in this interval.
For $x>x^*$, the function is concave and has exactly one fixed point.
So there is only $1$ fixed point larger than $x_0$. Contradiction.
The case that $f_d(x_1)<x_1$ is similar.

These two conditions could be rewritten as
\begin{align}
  \lambda>\frac{x_0(x_0+\gamma)^d}{(\beta x_0+1)^d}
  \label{eqn:unique-1}
\end{align}
and
\begin{align}
  \lambda<\frac{x_1(x_1+\gamma)^d}{(\beta x_1+1)^d}.
  \label{eqn:unique-2}
\end{align}
Notice that the right hand side has nothing to do with $\lambda$ in both \eqref{eqn:unique-1} and \eqref{eqn:unique-2}.

We want to see how conditions \eqref{eqn:unique-1} and \eqref{eqn:unique-2} change as $d$ changes.
Treat $d$ as a continuous variable.
Define
\begin{align*}
  g_i(d):=\frac{x_i(x_i+\gamma)^d}{(\beta x_i+1)^d}.
\end{align*}
where $i=0,1$ and $x_i$ is defined above depending on $\beta$, $\gamma$ and $d$.
Take the derivative:
\begin{align*}
  \frac{g_i'(d)}{g_i(d)}
  &=\frac{\partial x_i}{\partial d}\left(\frac{1}{x_i}+\frac{d}{x_i+\gamma}-\frac{d\beta}{\beta x_i+1}\right)+\log(x_i+\gamma)-\log(\beta x_i+1)\\
  &=\frac{\partial x_i}{\partial d}\left(\frac{1}{x_i}+\frac{d(1-\beta\gamma)}{(x_i+\gamma)(\beta x_i+1)}\right)+\log\frac{x_i+\gamma}{\beta x_i+1}\\
  &=\frac{\partial x_i}{\partial d}\left(\frac{1}{x_i}-\frac{1}{x_i}\right)+\log\frac{x_i+\gamma}{\beta x_i+1}=\log\frac{x_i+\gamma}{\beta x_i+1}.
\end{align*}

If $\beta\le 1$ these two functions are increasing in $d$.
Recall that $\critD=\frac{\sqrt{\beta\gamma}+1}{\sqrt{\beta\gamma}-1}$, and $\critLint=g_1(\ceil{\critD})=(\gamma/\beta)^{\frac{\ceil{\critD+1}}{2}}$.
Thus if $\lambda<\critLint$, \eqref{eqn:unique-2} holds for all integers $d$.
On the other hand, $x_0=o(1)$ by \eqref{eqn:x0x1:asymp}, and 
\begin{align*}
  g_0(d) &=\frac{x_0(x_0+\gamma)^d}{(\beta x_0+1)^d}=\frac{\gamma}{\beta x_1}\cdot \left( \frac{x_0+\gamma}{\beta x_0+1} \right)^{d} \\
  & > \frac{\gamma}{2\beta x^*}\cdot \left( \frac{x_0+\gamma}{\beta x_0+1} \right)^{d} 
  = \frac{\gamma}{d(\beta\gamma-1)-(\beta\gamma+1)} \cdot \left( \frac{x_0+\gamma}{\beta x_0+1} \right)^{d}\\
  &\rightarrow \infty \text{ as } d \text{ goes to }\infty.
\end{align*}
Hence there is no $\lambda$ such that \eqref{eqn:unique-1} holds for all integers $d$.

If $\beta>1$, then neither \eqref{eqn:unique-1} nor \eqref{eqn:unique-2} can hold for all integers $d$.
Since $x_0=o(1)$ by \eqref{eqn:x0x1:asymp}, 
similarly to the argument above, we have that
\begin{align*}
  g_0(d) &=\frac{x_0(x_0+\gamma)^d}{(\beta x_0+1)^d} >\frac{\gamma}{d(\beta\gamma-1)-(\beta\gamma+1)}\cdot \left( \frac{x_0+\gamma}{\beta x_0+1} \right)^{d}\\
  &\rightarrow \infty \text{ as } d \text{ goes to }\infty,
\end{align*}
which rules out \eqref{eqn:unique-1}.
Ruling out \eqref{eqn:unique-2} is completely analogous by noticing that $x_1\rightarrow\infty$ as $d$ goes to $\infty$ by \eqref{eqn:x0x1:asymp} and thus $g_1(d)\rightarrow 0$.
This proves Proposition \ref{prop:uniqueness:general}.

\subsection{Details about the Potential Method} \label{sec:missing-proofs}

In this section we provide missing details and proofs in Section \ref{sec:potential}.

To study correlation decay on trees, we use the standard recursion given in \eqref{eqn:tree-recursion}.
Recall that $T$ is a tree with root $v$.
Vertices $v_1,\ldots,v_d$ are $d$ children of $v$, and $T_i$ is the subtree rooted by $v_i$.
A configuration $\sigma_\Lambda$ is on a subset $\Lambda$ of vertices,
and $R^{\sigma}_T$ denote the ratio of marginal probabilities at $v$ given a partial configuration $\sigma$ on $T$.

We want to study the influence of another set of vertices, say $S$, upon $v$.
In particular, we want to study the range of ratios at $v$ over all possible configurations on $S$.
To this end, we define the lower and upper bounds as follows.
Notice that as $S$ will be fixed, we may assume that it is a subset of $\Lambda$.
\begin{definition}\label{definition-bounds}
  Let $T,v,\Lambda,\sigma_\Lambda,S,R^{\sigma}_T$ be as above.
  Define $R_v:=\min_{\tau_\Lambda}R^{\tau_\Lambda}_T$ and $R^v:=\max_{\tau_\Lambda}R^{\tau_\Lambda}_T$,
  where $\tau_\lambda$ can only differ from $\sigma_\Lambda$ on $S$.
  Define $\delta_v:=R^v-R_v$.
\end{definition}
Our goal is thus to prove that $\delta_v\le \exp(-\Omega(\mathrm{dist}(v,S)))$.
We can recursively calculate $R_v$ and $R^v$ as follows.
The base cases are:
\begin{enumerate}
  \item $v\in S$, in which case $R_v=0$ and $R^v=\infty$ and $\delta_v=\infty$;
  \item $v\in\Lambda\setminus S$,
    i.e.\ $v$ is fixed to be the same value in all $\tau_\Lambda$,
    in which case $R_v=R^v=0$ (or $\infty$) if $v$ is fixed to be blue (or green),
    and $\delta_v=0$;
  \item $v\not\in\Lambda$ and $v$ is the only node of $T$, in which case $R_v=R^v=\lambda$ and $\delta_v=0$.
\end{enumerate}
For $v\not\in\Lambda$,
since $F_d$ is monotonically increasing with respect to any $x_i$ for any $\beta\gamma>1$,
\begin{align*}
  R_v &= F_d(R_{v_1},...,R_{v_d}) \text{ and }  R^v = F_d(R^{v_1},...,R^{v_d}),
\end{align*}
where $R_{v_i}$ and $R^{v_i}$ are recursively defined
lower and upper bounds of $R_{T_i}^{\tau_\Lambda}$ for $1\le i\le d$.

Our goal is to show that $\delta_v$ decays exponentially in the depth of the recursion under certain conditions such as the uniqueness.
A straightforward approach would be to prove
that $\delta_v$ contracts by a constant ratio at each recursion step.
This is a sufficient, but not necessary condition for the exponential decay.
Indeed there are circumstances that $\delta_v$ does not necessarily decay in every step but does decay in the long run.
To amortize this behaviour, we use a \emph{potential function} $\Phi(x)$ and show that the correlation of a new recursion decays by a constant ratio.

To be more precise, the potential function $\Phi:\RR^+\rightarrow\RR^+$ is a differentiable and monotonically increasing function.
It maps the domain of the original recursion to a new one.
Let $y_i=\Phi(x_i)$.
We want to consider the recursion for $y_i$'s.
The new recursion function, which is the pullback of $F_d$, is defined as
\begin{align*}
  G_d(y_1,\dots,y_d):=\Phi(F_d(\Phi^{-1}(x_1),\dots,\Phi^{-1}(x_d))).
\end{align*}
The relationship between $F_d(\mathbf{x})$ and $G_d(\mathbf{y})$ is illustrated in Figure \ref{fig:commutative}.

\tikzcdset{row sep/normal=1.5cm,column sep/normal=1.5cm}
\begin{figure}[htpb]
  \centering
  \begin{tikzcd}
    \mathbf{x} \arrow[r,rightharpoonup,shift left=0.3ex,"\Phi"] \arrow[d, "F_d"]
    & \mathbf{y} \arrow[d, dashrightarrow, "G_d"] \arrow[l,rightharpoonup,shift left=0.3ex,"\Phi^{-1}"]\\
    F_d(\mathbf{x}) \arrow[r,rightharpoonup,shift left=0.3ex,"\Phi"]
    & G_d(\mathbf{y}) \arrow[l,rightharpoonup,shift left=0.3ex,"\Phi^{-1}"]
  \end{tikzcd}
  \caption[Commutative diagram between $F_d$ and $G_d$]{Commutative diagram between $F_d$ and $G_d$.}
  \label{fig:commutative}
\end{figure}

We want to prove Lemma \ref{lem:potential:algorithm} and Lemma \ref{lem:potential:algorithm:universal}.
To do so, we also define the upper and lower bounds of $y$.
Define $y_v=\Phi(R_v)$ and accordingly $y_{v_i}=\Phi(R_{v_i})$, for $1\le i\le d$,
as well as $y^v=\Phi(R^v)$ and $y^{v_i}=\Phi(R^{v_i})$, for $1\le i\le d$.
We have that
\begin{align}\label{eqn:ybounds}
  y_v&=G_d(y_{v_1},\dots,y_{v_d}) \text{ and } y^v=G_d(y^{v_1},\dots,y^{v_d}).
\end{align}
Let $\epsilon_v=y^v-y_v$.
For a good potential function, exponential decay of $\epsilon_v$ is sufficient to imply that of $\delta_v$.

\begin{lemma}  \label{lem:boundedness}
  Let $\Phi(x)$ be a good potential function for the field $\lambda$ at $v$.
  Then there exists a constant $C$ such that $\delta_v\le C \epsilon_v$ for any $\operatorname{dist}(v,S)\ge 2$.
\end{lemma}
\begin{proof}
  By \eqref{eqn:ybounds} and the Mean Value Theorem,
  there exists an $\widetilde{R}\in[R_v, R^v]$ such that
  \begin{align}\label{eqn:epsilon:delta}
    \epsilon_v &= \Phi(R^v)-\Phi(R_v) =\Phi'(\widetilde{R})\cdot\delta_v =\varphi(\widetilde{R})\cdot\delta_v.
  \end{align}
  Since $\operatorname{dist}(v,S)\ge 2$, we have that $R_v\ge \lambda\gamma^{-d}$ and $R^v\le\lambda\beta^d$.
  Hence $\widetilde{R}\in[\lambda\gamma^{-d},\lambda\beta^d]$, and by Condition \ref{cond:boundedness} of Definition \ref{def:potential},
  there exists a constant $C_1$ such that $\varphi(\widetilde{R})\ge C_1$.
  Therefore $\delta_v\le 1/C_1 \epsilon_v$.
\end{proof}

The next lemma explains Condition \ref{cond:contraction} of Definition \ref{def:potential}.

\begin{lemma}\label{lem:epsilon-contract}
  Let $\Phi(x)$ be a good potential function with contraction ratio $\alpha$.
  Then,
  \begin{align*}
    \epsilon_v\le \alpha\max_{1\le i\le d}\{\epsilon_{v_i}\}.
  \end{align*}
\end{lemma}
\begin{proof}
  First we use \eqref{eqn:ybounds}:
  \begin{align*}
    \epsilon_v &= y^v-y_v = G_d(y^{v_1},\dots,y^{v_d})-G_d(y_{v_1},\dots,y_{v_d}).
  \end{align*}
  Let $\mathbf{y}_1=(y^{v_1},\dots,y^{v_d})$ and $\mathbf{y}_0=(y_{v_1},\dots,y_{v_d})$.
  Let $\mathbf{z}(t)=t\mathbf{y}_1+(1-t)\mathbf{y}_0$ be a linear combination of $\mathbf{y}_0$ and $\mathbf{y}_1$ where $t\in[0,1]$.
  Then we have that
  \begin{align*}
    \epsilon_v & = G_d(\mathbf{z}(1))-G_d(\mathbf{z}(0)).
  \end{align*}
  By the Mean Value Theorem,
  there exist $\widetilde{t}$ such that $\epsilon_v=\deriv{G_d(\mathbf{z}(t))}{t}\Big|_{t=\widetilde{t}}$.
  Let $\widetilde{y_i}=\widetilde{t}y^{v_i} +(1-\widetilde{t})y_{v_i}$ for all $1\le i\le d$.
  Then we have that
  \begin{align}\label{eqn:epsilon_y}
    \epsilon_v &= \abs{\nabla G_d(\widetilde{y_1},\dots,\widetilde{y_d})\cdot(\epsilon_{v_1},\dots,\epsilon_{v_d})}.
  \end{align}
  It is straightforward to calculate that
  \begin{align}\label{eqn:partial_G}
    \frac{\partial G_d(\mathbf{y})}{\partial y_i} & =
    \frac{\varphi(F_d(\mathbf{R}))}{\varphi(R_i)}
    \cdot \frac{\partial F_d(\mathbf{R})}{\partial R_i},
  \end{align}
  where $R_i=\Phi^{-1}(y_i)$ and $\mathbf{y}$ and $\mathbf{R}$ are vectors composed by $y_i$'s and $R_i$'s.
  Plugging \eqref{eqn:partial_G} into \eqref{eqn:epsilon_y} we get that
  \begin{align*}
    \epsilon_v&= \varphi(F_d(\widetilde{\mathbf{R}}))\cdot
    \sum_{i=1}^d\abs{\frac{\partial F_d}{\partial R_i}}\frac{1}{\varphi(\widetilde{R_i})}\cdot \epsilon_{v_i}\\
    &\le \ctn{\varphi}{d}(\widetilde{R}_1,\ldots,\widetilde{R}_d)\cdot\max_{1\le i\le d}\{\epsilon_{v_i}\}\le \alpha\max_{1\le i\le d}\{\epsilon_{v_i}\},
  \end{align*}
  where $\widetilde{R_i}=\Phi^{-1}(\widetilde{y_i})$,
  $\widetilde{\mathbf{R}}$ is the vector composed by $\widetilde{R_i}$'s,
  and in the last line we use Condition \ref{cond:contraction} of Definition \ref{def:potential}.
\end{proof}

Note that the two conditions of a good potential function does not necessarily deal with all cases in the tree recursion.
At the root we have one more child than other vertices in a SAW tree.
Also, if $v$ has a child $u\in S$, then $\epsilon_u=\infty$ and the range in both conditions of Definition \ref{def:potential} does not apply.
To bound the recursion at the root, we have the following straightforward bound of the original recursion.

\begin{lemma}  \label{lem:delta-recursion}
  Let $(\beta,\gamma)$ be two parameters such that $\beta\gamma>1$ and $\beta<\gamma$.
  Let $v$ be a vertex and $v_i$ be its children for $1\le i\le d$.
  Suppose $\delta_{v_i}\le C$ for some $C>0$ and all $1\le i\le d$.
  Then,
  \begin{align*}
    \delta_v\le d \lambda_v(\beta\gamma-1)\gamma^{-1} \beta^d C.
  \end{align*}
\end{lemma}
\begin{proof}
  It is easy to see that $\gamma\ge 1$.
  By the same argument as in Lemma \ref{lem:epsilon-contract} and \eqref{eqn:tree-recursion}, there exists $x_i$'s such that
  \begin{align*}
    \delta_v &= \abs{\nabla F_d(x_1,\dots,x_d)\cdot(\delta_{v_1},\dots,\delta_{v_d})} \le C \sum_{i=1}^d \abs{\pderiv{F_d(\mathbf{x})}{x_i}},
  \end{align*}
  where $\mathbf{x}$ is the vector composed by $x_i$'s.
  Then, we have that
  \begin{align*}
    \abs{\pderiv{F_d(\mathbf{x})}{x_i}} & = \frac{d(\beta\gamma-1)F_d(\mathbf{x})}{(x_i+\gamma)(\beta x_i+1)}
    \le d \lambda_v(\beta\gamma-1)\gamma^{-1} \beta^d,
  \end{align*}
  where we use the fact that $F_d(\mathbf{x})\le \lambda_v\beta^d$ for any $x_i\in[0,\infty)$ and $\beta\gamma>1$.
  The lemma follows.
\end{proof}

Now we are ready to prove Lemma \ref{lem:potential:algorithm}.

\begin{proof}[Proof of Lemma \ref{lem:potential:algorithm}]
  Given $G$ and a partial configuration $\sigma_\Lambda$ on a subset $\Lambda\subseteq V$ of vertices,
  we first claim that we can approximate $p_v^{\sigma_\Lambda}$ within additive error $\epsilon$ deterministically
  in time $O\left(\epsilon^{\frac{\log\Delta}{\log\alpha}}\right)$.
  We construct the SAW tree $T=T_{\text{SAW}}(G,v)$.
  Due to Proposition \ref{prop:SAW},
  we only need to approximate $p_v^{\sigma_\Lambda}$ in $T$, with respect to $v$ and an arbitrary vertex set $S$.
  We will also use $\sigma_\Lambda$ to denote the configuration in $T$ on $\Lambda_{SAW}$.
  Let $S$ be the set of vertices whose distance to $v$ is larger than $t$,
  where $t$ is a parameter that we will specify later.
  Let $\delta_v$ be defined as in Definition \ref{definition-bounds}
  with respect to $T$, $v$, $\Lambda$, $\sigma_\Lambda$, and $S$.
  We want to show that $\delta_v=O(\lambda\alpha^{t})$.

  The maximum degree of $T$ is at most $\Delta$.
  Thus the root $v$ has at most $\Delta$ children in $T$, and any other vertex in $T$ has at most $\Delta-1$ children.
  Assume $v$ has $k\ge 1$ children as otherwise we are done.
  We may also assume that $v\not\in S$ and let $t=\mathrm{dist}(v,S)-1\ge 1$.
  We recursively construct a path $u_0=v$, $u_1$,\dots,$u_{l}$ of length $l\le t$ as follows.
  Given $u_i$, if there is no child of $u_i$,
  then we stop and let $l=i$.
  Otherwise $u_i$ has at least one child.
  If $i=t$ then we stop and let $l=t$.
  Otherwise $l<t$ and let $u_{i+1}$ be the child of $u_i$ such that $\epsilon_{u_{i+1}}$ takes the maximum $\epsilon$ among all children of $u_i$.
  In other words, by Lemma \ref{lem:epsilon-contract},
  we have that
  \begin{align}\label{eqn:recursive}
    \epsilon_{u_i}\le\alpha\epsilon_{u_{i+1}},
  \end{align}
  for all $1\le i\le l-1$.
  Notice that \eqref{eqn:recursive} may not hold for $i=0$
  since $v=u_0$ has possibly $\Delta$ children.

  First we note that for all $1\le i\le l$, $\text{dist}(v,u_i)=i\le l\le t$, and therefore $u_i\not\in S$.
  If we met any vertex $u_l$ with no child, then we claim that $\epsilon_{u_l}=0$.
  This is because $u_l$ is either a free vertex with no child or $u_l\in\Lambda$ but $u_l\not\in S$.
  However since $\epsilon_{u_l}$ takes the maximum $\epsilon$ among all children of $u_{l-1}$,
  we have that for all children of $u_{i-1}$, $\epsilon=0$, which implies that $\epsilon_{u_{i-1}}=0$.
  Recursively we get that $\epsilon_v=\epsilon_{u_0}=0$ and clearly the theorem holds by \eqref{eqn:epsilon:delta}.

  Hence we may assume that $l=t$.
  Since $u_l\not\in S$, we have that $\delta_{u_l}\le \lambda_{u_l}\beta^{-(\Delta-1)}$ if $\beta > 1$, or $\delta_{u_l}\le \lambda_{u_l}$ if $\beta \le 1$.
  Hence by \eqref{eqn:epsilon:delta} and Condition \ref{cond:boundedness} in Definition \ref{def:potential},
  we have that $\epsilon_{u_l}\le C_0$ for some constant $C_0$.
  Applying \eqref{eqn:recursive} inductively we have that
  \begin{align*}
    \epsilon_{u_1}\le\alpha^l\epsilon_{u_l}\le\alpha^{t} C_0.
  \end{align*}
  Hence by Lemma \ref{lem:boundedness},
  we there exists another constant $C_1$ such that $\delta_{u_1}\le \alpha^{t} C_1$.
  To get a bound on $\delta_{u_0}$, we use Lemma \ref{lem:delta-recursion}, which states that
  \begin{align*}
    \delta_{u_0}\le d_0 \lambda_v(\beta\gamma-1)\gamma^{-1} \beta^{d_0}\delta_{u_1}\le d_0 \lambda_v(\beta\gamma-1)\gamma^{-1} \beta^{d_0}\alpha^{t} C_1
    =O(\lambda \alpha^t),
  \end{align*}
  where $d_0\le\Delta$ is the degree of $v=u_0$.

  Hence the recursive procedure returns $R_v$ and $R^v$
  such that $R_v\le R_T^{\sigma_\Lambda}\le R^v$,
  and $R^v-R_v=O(\lambda\alpha^t)$ where $\alpha<1$ is the contraction ratio.
  Note that $R_T^{\sigma_\Lambda}=R_{G,v}^{\sigma_\Lambda}=\frac{p_v^{\sigma_\Lambda}}{1-p_v^{\sigma_\Lambda}}$.
  Let $p_0=\frac{R_v}{R_v+1}$ and $p_1=\frac{R^v}{R^v+1}$.
  Then $p_0\le p_v^{\sigma_\Lambda}\le p_1$ and
  \begin{align}
    p_1-p_0=\frac{R^v}{R^v+1}-\frac{R_v}{R_v+1}\le R^v-R_v=O(\lambda\alpha^t). \label{eqn:p-bound}
  \end{align}
  The recursive procedure runs in time $O(\Delta^t)$
  since it only needs to construct the first $t$ levels of the self-avoiding walk tree.
  For any $\epsilon>0$, let $t=O(\log_{\alpha}\epsilon-\log_{\alpha}\lambda)$ so that $R^v-R_v < \epsilon$.
  This gives an algorithm which approximates $p_v^{\sigma_\Lambda}$
  within an additive error $\epsilon$ in time $O\left(\left(\frac{\epsilon}{\lambda}\right)^{\frac{\log\Delta}{\log\alpha}}\right)$.

  Then we use self-reducibility to reduce computing $Z_{\beta,\gamma,\pi}(G)$ to computing conditional marginal probabilities.
  To be specific, let $\sigma$ be a configuration on a subset of $V$ and $\tau$ be sampled according to the Gibbs measure.
  Let $p_v^{\sigma}:=\Pr\left(\tau(v)=1\mid \sigma\right)$ be the conditional marginal probability.
  We can compute $Z_{\beta,\gamma,\pi}(G)$ from $p_v^{\sigma}$ by the following standard procedure.
  Let $v_1,\ldots,v_n$ enumerate vertices in $G$.
  For $0\le i\le n$, let $\sigma_i$ be the configuration fixing the first $i$ vertices $v_1,\ldots,v_i$ as follows:
  $\sigma_i(v_j)=\sigma_{i-1}(v_j)$ for $1\le j\le i-1$
  and $\sigma_i(v_i)$ is fixed to the spin $s$
  so that $p_{i}:=\Pr\left( \tau(v_i)=s\mid \sigma_{i-1} \right)\ge 1/3$.
  This is always possible because clearly
  \begin{align*}
    \Pr\left( \tau(v_i)=0\mid \sigma_{i-1} \right) + \Pr\left( \tau(v_i)=1\mid \sigma_{i-1} \right)=1.
  \end{align*}
  In particular, $\sigma_n\in\{0,1\}^V$ is a configuration of $V$.
  The Gibbs measure of $\sigma_n$ is $\rho(\sigma_n)=\frac{w(\sigma_n)}{Z_{\beta,\gamma,\pi}(G)}$.
  On the other hand, we can rewrite $\rho(\sigma_n)=p_1 p_2\cdots p_n$ by conditional probabilities.
  Thus $Z_{\beta,\gamma,\pi}(G)=\frac{w(\sigma_n)}{p_1p_2\cdots p_n}$.
  The weight $w(\sigma_n)$ given in \eqref{eqn:weight} can be computed exactly in time polynomial in $n$.
  Note that $p_i$ equals to either $p_{v_i}^{\sigma_{i-1}}$ or $1-p_{v_i}^{\sigma_{i-1}}$.
  Since we can approximate $p_v^{\sigma_\Lambda}$ within an additive error $\epsilon$
  in time $O\left(\left(\frac{\epsilon}{\lambda}\right)^{\frac{\log\Delta}{\log\alpha}}\right)$,
  the configurations $\sigma_i$ can be efficiently constructed,
  which guarantees that all $p_i$'s are bounded away from 0.
  Thus the product $p_1p_2\cdots p_n$ can be approximated within a factor of $(1\pm n\epsilon')$
  in time $O\left(n\left(\frac{\epsilon'}{\lambda}\right)^{\frac{\log\Delta}{\log\alpha}}\right)$.
  Now let $\epsilon'=\frac{\epsilon}{n}$.
  We get the claimed FPTAS for $Z_{\beta,\gamma,\pi}(G)$.
\end{proof}

Lemma \ref{lem:potential:spatial-mixing} follows almost immediately
from Lemmas \ref{lem:boundedness}, \ref{lem:epsilon-contract}, and \ref{lem:delta-recursion} as in the proof above.
The only issue is that the range of $x$ should be restricted to $(0,\lambda]$.
This is guaranteed by Claim \ref{claim:bound}.

Finally we show Lemma \ref{lem:potential:algorithm:universal}.

\begin{proof}[Proof of Lemma \ref{lem:potential:algorithm:universal}]
  By the same proof of Lemma \ref{lem:potential:algorithm},
  we only need to approximate the marginal probability at the root $v$ of a tree $T$.
  By Condition \ref{cond:universal-contraction} of Definition \ref{def:potential:universal},
  $\ctn{\varphi}{d}(x_1,\cdots,x_d)<\alpha^{\ceil{\log_M (d+1)}}$.
  Denote by $B(\ell)$ the set of all vertices whose $M$-based depths of $v$ is at most $\ell$ in $T$.
  Hence $|B(\ell)|\le M^\ell$.
  Let $S=\{u\mid \operatorname{dist}(u,B(\ell))>1\}$, which is essentially the same $S$ as in Lemma \ref{lem:potential:algorithm}, but under a different metric.
  We can recursively compute upper and lower bounds $R^v$ and $R_v$ of $R_T^{\sigma_\Lambda}$ such that $R_v\le R_T^{\sigma_\Lambda} \le R^v$,
  with the base case that for any vertex $u\in S$ trivial bounds $R_u=0$ and $R^u=\infty$ are used.

  We proceed as in the proof of Lemma \ref{lem:potential:algorithm}.
  Without loss of generality,
  we construct a path $u_0u_1\cdots u_{k}$ in $T$ from the root $u_0=v$ to a $u_k$
  with $\ell_M(u_{k-1})\le\ell$ and $\ell_M(u_{k})>\ell$.
  As in the proof of Lemma \ref{lem:epsilon-contract},
  $\epsilon_{u_j}\le C^{\varphi}_{d_j}(x_{j,1},\ldots,x_{j,d_j})\cdot\epsilon_{u_{j+1}}$for all $0\le j\le k-1$,
  where $d_j$ is the number of children of $u_j$ and $x_{j,i}\in[0,\infty)$, $1\le i\le d_j$.
  Hence we have that
  \begin{align*}
    \epsilon_v &\le \epsilon_{u_k}\cdot\prod_{j=0}^{k-1}\alpha^{\lceil\log_M (d_j+1)\rceil}
    \le \epsilon_{u_k}\cdot\alpha^{\sum_{j=0}^{k-1}\lceil\log_M (d_j+1)\rceil}\\
    &=\epsilon_{u_k}\cdot\alpha^{\ell_M(u_k)}
    \le\epsilon_{u_k}\cdot\alpha^\ell.
  \end{align*}
  Note that $\operatorname{dist}(u_{k},B(\ell))=1$ and hence $u_k\not\in S$.
  So $\delta_{u_k}<\lambda_{u_k}\le \lambda$.
  By \eqref{eqn:epsilon:delta}, we have that $\epsilon_{u_k}\le \varphi(\widetilde{R}) \delta_{u_k}$,
  for some $\widetilde{R}\in[\lambda_{u_k}\gamma^{-d_k},\lambda_{u_k}\beta^{d_k}]$.
  Hence $\epsilon_{u_k}<C_2\lambda$ by Condition \ref{cond:universal-boundedness} of Definition \ref{def:potential:universal},
  and $\epsilon_v < \lambda \alpha^{\ell} C_2$.
  By \eqref{eqn:epsilon:delta} and Condition \ref{cond:universal-boundedness} of Definition \ref{def:potential:universal} again,
  we have that $\delta_v\le \lambda \alpha^{\ell} C_2/C_1$.

  The rest of the proof goes the same as that of Lemma \ref{lem:potential:algorithm}.
  The running time has an extra $n^2$ factor since we need to go down two more levels (in the worst case) outside of $B(\ell)$.
\end{proof}

\subsection{Proofs of Lemma \ref{lem:critL} and Lemma \ref{lem:inequality-key}} \label{sec:lemma-proof}

In this section we show Lemma \ref{lem:critL} and Lemma \ref{lem:inequality-key}.
We prove Lemma \ref{lem:critL} first, and then use it to show Lemma \ref{lem:inequality-key}.

\begin{proof}[Proof of Lemma \ref{lem:critL}]
  It is trivial if $\beta\leq 1$.
  Now assume that $\beta > 1$.
  As $\frac{\beta x+1}{x+\gamma}$ is increasing in $x$, it is equivalent to show that
  \[\frac{\gamma-1}{\beta -1} \geq \critL =\left(\frac{\gamma}{\beta}\right)^{\frac{\sqrt{\beta \gamma}}{\sqrt{\beta \gamma}-1}}
    \ \ \ \ \ \ \Leftrightarrow \ \ \ \ \ \ \ \ \
    \log(\gamma-1)- \log (\beta -1) \geq \frac{\sqrt{\beta \gamma}}{\sqrt{\beta \gamma}-1} \log \left(\frac{\gamma}{\beta}\right) .\]
  Let $\gamma= k^2 \beta$ with $k\ge 1$.
  We only need to show that $r(k)\ge 0$ for $k\ge 1$, where $r(k)$ is defined as
  \[r(k):= \log(\beta k^2 -1)- \log (\beta -1) - \frac{2 \beta k }{\beta k-1}\log k.\]
  Since $r(1)=0$, it is enough to prove that $r(k)$ is increasing for $k\ge 1$.
  It can be easily verified as
  \begin{align*}
    r'(k) & = \frac{2 \beta k}{\beta k^2 -1} -\frac{2 \beta}{\beta k-1} + \frac{2 \beta}{(\beta k-1)^2} \log k\\
          & = \frac{2 \beta}{(\beta k-1)^2 (\beta k^2 -1)}\left( (\beta k^2 -1) \log k - (k-1)(\beta k -1)\right).
  \end{align*}
  So, it is sufficient to show that
  \[ (\beta k^2 -1) \log k - (k-1)(\beta k -1) \geq 0. \]
  Since $k\geq 1$, we have that $\log k \geq 1-\frac{1}{k}$.
  It implies that
  \[ (\beta k^2 -1) \log k - (k-1)(\beta k -1) \geq (\beta k^2 -1) (1-\frac{1}{k})  - (k-1)(\beta k -1)= \frac{(k-1)^2}{k}\geq 0 . \]
  This completes the proof.
\end{proof}

Then we show Lemma \ref{lem:inequality-key}.

\begin{proof}[Proof of Lemma \ref{lem:inequality-key}]
  Let $g(x):=(\beta \gamma-1)x \log{\frac{\lambda_c}{x}} - (\beta x+1)(x+\gamma) \log{\frac{x+\gamma}{\beta x+1}}$.
  Hence it is equivalent to show that $g(x)\le 0$ for all $0<x<\critL$.
  Take the derivative of $g(x)$ and we have that
  \begin{align*}
    g'(x) 
    & =(\beta \gamma-1)(\log{\frac{\lambda_c}{x}} -1) - (2 \beta x+ \beta \gamma+1) \log{\frac{x+\gamma}{\beta x+1}}\\
    &\ \ \ - (\beta x+1)(x+\gamma) \left(\frac{1}{x+\gamma} -\frac{\beta}{\beta x+1}\right)\\
    &=(\beta \gamma-1)\log{\frac{\lambda_c}{x}}-(2 \beta x+ \beta \gamma+1) \log{\frac{x+\gamma}{\beta x+1}}.
  \end{align*}
  By direct calculation, $g\left(\sqrt{\frac{\gamma}{\beta}}\right)=0$ and $g'\left(\sqrt{\frac{\gamma}{\beta}}\right)=0$.
  Then we prove \eqref{eqn:inequality-key} for the case of $0<x<\sqrt{\frac{\gamma}{\beta}}$ and $\sqrt{\frac{\gamma}{\beta}}<x<\critL$ separately.

  If $0<x<\sqrt{\frac{\gamma}{\beta}}$, it is sufficient to verify that $g'(x)>0$.
  We only need to show that $g'(x)$ is decreasing since  $g'\left(\sqrt{\frac{\gamma}{\beta}}\right)=0$.
  It is easily verified by taking the derivative again:
  \begin{align*}
    g''(x)
    & =-\frac{\beta \gamma-1}{x} - 2 \beta  \log{\frac{x+\gamma}{\beta x+1}} - (2 \beta x+ \beta \gamma+1) \left(\frac{1}{x+\gamma} -\frac{\beta}{\beta x+1}\right)\\
    &=- 2 \beta  \log{\frac{x+\gamma}{\beta x+1}} -(\beta \gamma-1) \left(\frac{1}{x}- \frac{2 \beta x+ \beta \gamma+1}{(x+\gamma)(\beta x+1)}\right)\\
    &=- 2 \beta  \log{\frac{x+\gamma}{\beta x+1}} -(\beta \gamma-1) \frac{\gamma - \beta x^2}{x(x+\gamma)(\beta x+1)}<0,
  \end{align*}
  where the last inequality uses the fact that $\frac{x+\gamma}{\beta x+1}\ge 1$ by Lemma \ref{lem:critL} and $x<\sqrt{\frac{\gamma}{\beta}}$.

  If $\sqrt{\frac{\gamma}{\beta}}<x<\critL$, then we show \eqref{eqn:inequality-key} directly.
  First notice that as $x\neq \sqrt{\frac{\gamma}{\beta}}$,
  \begin{align*}
    \frac{x}{(\beta x+1)(x+\gamma) } = \frac{1}{\beta x  +\frac{\gamma}{x} + \beta \gamma +1 } < (\sqrt{\beta \gamma}+1)^{-2},
  \end{align*}
  Given this, in order to get (\ref{eqn:inequality-key}), it is sufficient to show that $h(x) <0$ where
  \begin{align*}
    h(x):=\frac{\sqrt{\beta \gamma}-1}{\sqrt{\beta \gamma}+1}\log{\frac{\lambda_c}{x}} - \log{\frac{x+\gamma}{\beta x+1}}.
  \end{align*}
  In fact, $h(x)$ is a decreasing function as
  \begin{align*}
    h'(x) & =-\frac{\sqrt{\beta \gamma}-1}{x(\sqrt{\beta \gamma}+1)}- \frac{1}{x+\gamma} +  \frac{\beta}{\beta x+1} \\
    &=-\frac{(\sqrt{\beta \gamma}-1) \left((x+\gamma)(\beta x+1) - (\sqrt{\beta \gamma}+1)^2 x\right)}{x(\sqrt{\beta \gamma}+1)(x+\gamma)(\beta x+1)} \\
    &=-\frac{(\sqrt{\beta \gamma}-1) \left(\sqrt{\beta} x - \sqrt{\gamma}\right)^2}{x(\sqrt{\beta \gamma}+1)(x+\gamma)(\beta x+1)} \le 0.
  \end{align*}
  Notice that $h\left(\sqrt{\frac{\gamma}{\beta}}\right)=0$.
  It implies that $h(x)<0$ for all $x> \sqrt{\frac{\gamma}{\beta}}$.
  This completes the proof.
\end{proof}

\section*{Acknowledgements}
A preliminary version \cite{GL16} has appeared in RANDOM 2016.

We thank Liang Li, Jingcheng Liu, and Chihao Zhang for stimulating discussion.
In particular, the example of the $5$-$7$ tree in Section \ref{sec:hardness} is an outcome from such discussion.
We also thank organizers of the ``IMA-GaTech Workshop on the Power of Randomness in Computation'' in March 2015.
The current work stems from discussion during the workshop.

We also thank anonymous referees for their meticulous check over the draft version.

\bibliographystyle{alpha}
\bibliography{bib}

\end{document}